\documentclass[11pt,letterpaper]{article}
\usepackage{comment}
\usepackage{url,nicefrac,bm}
\usepackage[ruled,vlined,linesnumbered]{algorithm2e}
\usepackage{thmtools,thm-restate,amsmath,amssymb,amsthm,mathtools}

\usepackage[margin=0.88in]{geometry}

\usepackage{paralist}
\usepackage{xcolor}

\newtheorem{lemma}{Lemma}[section]
\newtheorem{proposition}{Proposition}[section]
\newtheorem{theorem}{Theorem}[section]
\newtheorem{corollary}{Corollary}[section]

\newtheorem{definition}{Definition}[section]
\newtheorem{remark}[theorem]{Remark}

\newcommand{\vecp}{{\mathbf p}}

\newcommand{\vecx}{{\mathbf x}}

\newcommand{\B}{{B}}
\newcommand{\G}{{G}}
\newcommand{\tB}{{\hat{B}}}
\newcommand{\tG}{{\hat{G}}}

\newcommand{\p}{{\mathbf p}}
\newcommand{\CG}{{\mathcal G}}
\newcommand{\CM}{{\mathcal M}}

\newcommand{\0}{{\mathbf 0}} 
\newcommand{\f}{{\mathbf f}}
\newcommand{\x}{{\mathbf x}}

\newcommand{\bc}{\bar{c}}

\newcommand{\ot}{\leftarrow}
\newcommand{\D}{\displaystyle}

\newcommand{\eps}{\varepsilon}

\newcommand{\classP}{{\sf P}}
\newcommand{\classNP}{{\sf NP}}
\newcommand{\classAPX}{{\sf APX}}
\newcommand{\classCLS}{{\sf CLS}}
\newcommand{\classPLS}{{\sf PLS}}
\newcommand{\classPPAD}{{\sf PPAD}}

\let\oldnl\nl
\newcommand{\nonl}{\renewcommand{\nl}{\let\nl\oldnl}}

\title{Satiation in Fisher Markets and Approximation of Nash Social Welfare}
\author{Jugal Garg\thanks{University of Illinois at Urbana-Champaign. Supported by NSF CRII Award 1755619.}\\
\texttt{jugal@illinois.edu}
\and 
Martin Hoefer\thanks{Institut f\"ur Informatik, Goethe-Universit\"at Frankfurt/Main} \\
\texttt{mhoefer@cs.uni-frankfurt.de}
\and
Kurt Mehlhorn\thanks{Max-Planck-Institut f\"ur Informatik}\\ 
\texttt{mehlhorn@mpi-inf.mpg.de}}
\date{}

\begin{document}
\maketitle

\begin{abstract}
We study linear Fisher markets with satiation. In these markets, sellers have \emph{earning limits} and buyers have \emph{utility limits}. Beyond natural applications in economics, these markets arise in the context of maximizing Nash social welfare when allocating indivisible items to agents. In contrast to markets with \emph{either} earning \emph{or} utility limits, markets with both limits have not been studied before. They turn out to have fundamentally different properties. 

In general, existence of competitive equilibria is not guaranteed. We identify a natural property of markets (termed \emph{money clearing}) that implies existence. We show that the set of equilibria is not always convex, answering a question of \cite{ColeDGJMVY17}. We design an FPTAS to compute an approximate equilibrium and prove that the problem of computing an exact equilibrium lies in the intersection of complexity classes \classPLS\ 
and \classPPAD. 
For a constant number of buyers or goods, we give a polynomial-time algorithm to compute an exact equilibrium.

We show how (approximate) equilibria can be rounded and provide the first constant-factor approximation algorithm (with a factor of $2.404$) for maximizing Nash social welfare when agents have budget-additive valuations. Finally, we significantly improve the approximation hardness for additive valuations to $\sqrt{8/7} > 1.069$ (over 1.00008 in~\cite{Lee17}).
\end{abstract}

\newcommand{\tU}{\tilde{U}}
\newcommand{\tu}{\tilde{u}}
\newcommand{\tCM}{\tilde{\mathcal M}}

\newcommand{\gijk}{\gamma_{ijk}}
\newcommand{\lk}{{\lambda_k}}
\newcommand{\tli}{{\tilde{\lambda}_i}}
\newcommand{\Li}{{\Lambda_i}}

\newcommand{\ai}{\alpha_i}
\newcommand{\ak}{\alpha_k}
\newcommand{\tai}{\tilde{\alpha}_i}

\newcommand{\1}{{\mathbf 1}}

\renewcommand{\P}{{\p^{min}}}
\renewcommand{\L}{{\sf LCP }}

\section{Introduction}

The Fisher market model was introduced by Irving Fisher in 1891~\cite{BrainardS00}, and it has been a prominent approach to study market properties ever since. A Fisher market consists of a set of buyers and a set of divisible goods. Buyers come to the market with money and have utility functions over allocations of goods. We assume that each good is brought by a seller and comes in unit supply. Given prices of goods, each buyer demands a bundle of goods that maximizes her utility. In a market equilibrium, the prices are such that all goods are fully sold. 

Market equilibrium is a central solution concept in economics that has found many surprising applications even in \emph{non-market settings}, which do not involve an exchange of money. The reason is that market equilibria exhibit remarkable fairness and efficiency properties -- the most prominent example is the popular fairness notion of \emph{competitive equilibrium with equal incomes} (CEEI)~\cite{Moulin03}. 

In a linear Fisher market model, each buyer has a linear utility function. Linear market models have been extensively studied since 1950s~\cite{EisenbergG59,Gale60}. Recently, two natural generalizations of linear Fisher markets, based on satiation, are introduced to model real-life preferences. In these models either $i)$ buyers have utility limits or $ii)$ sellers have earning limits.

In the first model, each buyer has an upper limit on the amount of utility that they want to derive. Each buyer spends the least amount of money to purchase a bundle of goods that maximizes their utility up to the limit. She takes back any unused part of her money. These utility functions with limits are also known as budget-additive utility functions, which arise naturally in cases where agents have an intrinsic upper bound on their utility. For example, if the goods are food and the utility of a food item for a particular buyer is its calorie content, calories above a certain threshold do not increase the utility of the buyer. In addition, there are a variety of further applications in adword auctions and revenue maximization problems; see~\cite{BeiGHM16} for details. Equilibrium in this model always exists and can be captured by a convex program~\cite{ColeDGJMVY17,BeiGHM16}, and there is a combinatorial polynomial-time algorithm to find it~\cite{BeiGHM16}. 

In the second model, each seller has an upper limit on the amount of money that they want to earn. Each seller sells the least amount of their good to earn the maximum amount of money up to the limit. He takes back any unsold portion of the good. This is a natural property in many settings, e.g., when sellers have revenue targets; see~\cite{ColeDGJMVY17} for further applications. Equilibrium in this model may not always exist. However, a necessary and sufficient condition for the existence, a convex programming formulation, and combinatorial polynomial-time algorithms for computation of equilibrium when it exists have been obtained~\cite{ColeDGJMVY17,BeiGHM17}. 

The natural generalization of the two models, where both buyers and sellers have limits, was only briefly introduced in~\cite{ColeDGJMVY17}. The authors posed an intriguing open question of obtaining a convex programming formulation for equilibria in this model. Markets with both utility and earning limits are the main subject of our paper. We study markets of this class that satisfy the sufficient condition for existence in the context of earning limits. In this class of markets, we show that the set of equilibria can be non-convex, thereby answering the question of~\cite{ColeDGJMVY17}. We design an FPTAS to compute an approximate equilibrium and prove that the problem of computing an exact equilibrium lies in the intersection of the complexity classes \classPLS\ (Polynomial Local Search) and \classPPAD\ (Polynomial Parity Arguments on Directed Graphs). For a constant number of buyers or goods, we give a polynomial time algorithm to compute an exact equilibrium. To the best of our knowledge, this is the first market equilibrium problem which lies in \classPPAD\ $\cap$ \classPLS\ and for which no polynomial-time algorithm is known. 

Beyond the economic interest in modeling buyers and sellers' preferences, generalized Fisher markets have also found further applications, especially for approximating the maximum Nash social welfare when allocating \emph{indivisible} items to a set of agents. The Nash social welfare is defined as the geometric mean of agents' valuations, which provides an interesting trade-off between the extremal objectives of social welfare and max-min welfare. In social welfare, the objective is to maximize the \emph{sum of valuations}, while in max-min welfare, the objective is to maximize the \emph{minimum of valuations}. The Nash social welfare objective has been proposed in the classic game theory literature by Nash~\cite{Nash50} when solving the bargaining problem. It is closely related to the notion of proportional fairness studied in networking~\cite{Kelly97}. Nash social welfare satisfies a set of desirable axioms such as independence of unconcerned agents, the Pigou-Dalton transfer principle, and independence of common utility scale (see, e.g.,~\cite{KanekoN79, Moulin88}). The latter implies that, in contrast to both social welfare and max-min fairness, it is invariant to individual scaling of each agent valuation with a possibly different constant factor. 

The problem of maximizing the Nash social welfare objective is known to be \classAPX-hard~\cite{Lee17}, even for additive valuations. As a remarkable result, Cole and Gkatzelis~\cite{ColeG18} gave the first constant-factor approximation algorithm for additive valuations. The constant was subsequently improved to 2~\cite{ColeDGJMVY17}. The algorithm computes and rounds an equilibrium of a Fisher market where sellers have earning limits. Moreover, the approach has been extended to provide a 2-approximation in multi-unit markets with agent valuations, which remain additive-separable over items~\cite{BeiGHM17}, but might be concave in the number of copies received for each item~\cite{AnariMGV18}.

In this paper, we show that (approximate and exact) market equilibria that are computed by our algorithms can be rounded to give the first constant-factor approximation algorithm for maximizing the Nash social welfare when agents have budget-additive valuation functions. These are \emph{a class of non-separable submodular} valuation functions. The analysis of budget-additive valuations significantly advances our understanding beyond additive and towards submodular ones. Moreover, budget-additive valuations are of interest in a variety of applications, most prominently in online advertising~\cite{MehtaSVV07, Mehta12}. They have been studied frequently in the literature, e.g., for offline social welfare maximization~\cite{AndelmanM04, AzarBKMN08, Srinivasan08, ChakrabartyG10, Kalaitzis16}, online algorithms~\cite{BuchbinderJN07,DevanurJSW19}, mechanism design~\cite{BuchfuhrerDFKMPSSU10}, Walrasian equilibrium~\cite{RoughgardenT15,FeldmanGL16}, and market equilibrium~\cite{BeiGHM16, ColeDGJMVY17}.

Finally, we also strengthen the existing hardness results for approximating Nash social welfare. We provide a new inapproximability bound of $1.069$ that applies even in the case of \emph{additive} valuations. This significantly improves the constant over 1.00008 in~\cite{Lee17}.

\subsection{Contribution and Techniques}\label{sec:technical}
\paragraph{Money-Clearing Markets and an FPTAS}

We study Fisher markets with additive valuations that have earning and utility limits. In markets with utility limits, a market equilibrium always exists. For markets with earning limits, equilibria exist if and only if the market satisfies a natural condition on budgets and earning limits (which we term \emph{money clearing}). This condition holds, in particular, for all market instances that arise in the context of computing approximate solutions for the Nash social welfare problem. In both markets models with \emph{either} earning \emph{or} utility limits, the set of equilibria is always convex. 

For \emph{both} earning \emph{and} utility limits, we also concentrate on markets with the money-clearing condition, which we show is sufficient (but not necessary) for the existence of an equilibrium. We prove that the set of market equilibria can be \emph{non-convex}. Hence, in contrast to the above cases, the toolbox for solving convex programs (e.g., ellipsoid~\cite{ColeDGJMVY17} or scaling algorithms~\cite{ColeG18,BeiGHM16}) is not directly applicable for computing an equilibrium.

Our main result is a new algorithm to compute an approximate equilibrium. Based on a constant $\eps > 0$, it perturbs the valuations and rounds the parameters $v_{ij}$ up to the next power of $(1+\eps)$. Then it computes an exact equilibrium of the perturbed market in polynomial time, which represents an approximate equilibrium in the original market. This yields a novel FPTAS for markets with earning and utility limits. We note that the non-convexity of equilibria also applies to perturbed markets, which is surprising since we show an exact polynomial-time algorithm for computing an equilibrium.

To compute an exact equilibrium in the perturbed market, we first obtain an equilibrium (prices $\p$, allocation $\x$) of a market that results from ignoring all utility limits~\cite{ColeDGJMVY17,BeiGHM17}. This is not an equilibrium of the market with both limits, because some buyers may be overspending. Let the \emph{surplus of a buyer} be the money spent minus the money needed to earn the optimal utility, and similarly let the \emph{surplus of a good} be the target earning minus the actual earning. Let $S$ be the set of buyers who have positive surplus at prices $\p$. Our idea is to pick a buyer, say $k$, in $S$ and decrease the prices of goods in a coordinated fashion. The goal is to make $k$'s surplus zero while maintaining the surpluses of all goods and all buyers not in $S$ to be zero. We show that after a polynomial number of iterations of price decrease, either the surplus of buyer $k$ becomes zero or we discover a good with price 0 in equilibrium. Picking a particular buyer is crucial in the analysis, because we rely on this buyer to show that a certain parameter strictly decreases. This implies substantial price decrease of goods and polynomial running time.

\paragraph{Complexity of Exact Equilibria} 
In addition to the FPTAS, we examine the complexity of computing an exact equilibrium in money-clearing markets. We show that this problem lies in $\classPPAD \cap \classPLS$. For membership in \classPLS\ we first design a finite-time algorithm to compute an exact equilibrium. We define a finite \emph{configuration} space such that the algorithm proceeds through a sequence of configurations. We show that configurations in the sequence do not repeat and the algorithm terminates with an equilibrium. By defining a suitable potential function over configurations, we show that the problem is in \classPLS. As a refinement, for a constant number of buyers or sellers, we show that the number of configurations is polynomially bounded using a cell decomposition technique. This implies that our algorithm computes an equilibrium in polynomial time if the number of buyers or goods is constant. 

For membership in \classPPAD\ we first derive a formulation as a linear complementarity problem (LCP). It captures all equilibria, but it also has non-equilibrium solutions. To discard the non-equilibrium solutions, we incorporate a positive lower bound on several variables. This turns out to be a non-trivial adjustment, because a subset of prices may be zero at all equilibria, so we must be careful not to discard equilibrium solutions. Then, we suitably add an auxiliary variable to the LCP and apply Lemke's algorithm~\cite{CottlePS92}. Under the money clearing condition, we can show that the algorithm is guaranteed to converge to an exact equilibrium. This, with a result of Todd~\cite{Todd76}, proves the problem lies in \classPPAD. 

\paragraph{Approximating Nash Social Welfare}
Finally, we consider the problem of maximizing Nash social welfare when allocating indivisible items to agents. We design an approximation algorithm that computes an equilibrium in a money-clearing market and rounds it to an integral allocation. Here we study the problem for agents with budget-additive valuation functions. For these instances, money-clearing markets with earning and utility limits represent a natural fractional relaxation. We use our algorithms to compute an exact equilibrium (in the FPTAS with respect to perturbed valuations). Given an exact equilibrium (with respect to either perturbed or original valuations), we provide a rounding algorithm that turns the fractional allocation into an integral one. While the algorithm exploits a tree structure of the equilibrium allocation as in~\cite{ColeG18}, the rounding becomes much more challenging, and we must be careful to correctly treat agents that reach their utility limits in the equilibrium. In particular, we first conduct several initial assignment steps to arrive at a solution where we have a set of rooted trees on agents and items, and each item $j$ has exactly one child agent $i$ who gets at least half of its fractional valuation from $j$. In the main step of the rounding algorithm, we need to ensure that the root agent $r$ receives one of its child items. Here we pick a child item $j$ that generates the most value for $r$. A problem arises at the child agent $i$ of $j$, since $r$ receiving $j$ could decrease $i$'s valuation by a lot more than a factor of 2. Recursively, we again need to enforce an allocation for the root agent, thereby ``stealing'' fractional value from one of its grandchildren agents. This approach may seem hopeless to yield any reasonable approximation guarantee, but we show that overall the agents only suffer by a small constant factor.

Our analysis of this rounding procedure provides a lower bound on the Nash social welfare obtained by the algorithm, which is complemented with an upper bound on the optimum solution. Both bounds crucially exploit the properties of agents (goods) that reach the utility (earning) limits in the market equilibrium. These bounds imply an approximation factor of $2e^{1/(2e)} < 2.404$. Since the equilibrium conditions apply with respect to perturbed valuations, we obtain a $(2e^{1/(2e)}+\varepsilon)$-approximation in polynomial time, for any constant $\varepsilon > 0$.

In terms of lower bounds, we strengthen the inapproximability bound to $\sqrt{8/7} > 1.069$. Our improvement is based on a construction for hardness of social welfare maximization for budget-additive valuations from~\cite{ChakrabartyG10}. For the Nash social welfare objective, we observe how to drop the utility limits and apply the construction even for additive valuations.

A preliminary version of this paper appeared at the 29th ACM-SIAM Symposium on Discrete Algorithms (SODA 2018)~\cite{GargHM18}.

\subsection{Related Work}

\paragraph{Market Equilibria} The problem of computing market equilibria is an intensely studied problem, so we restrict to previous work that appears most relevant. 

For linear Fisher markets, equilibria are captured by the Eisenberg-Gale convex program~\cite{EisenbergG59}. Later, Shmyrev~\cite{Shmyrev09} obtained another convex program for this problem. Cole et al.~\cite{ColeDGJMVY17} provide a dual connection between these and other convex programs. A combinatorial polynomial-time algorithm for computing an equilibrium in this model was obtained by Devanur et al.~\cite{DevanurPSV08}. Orlin~\cite{Orlin10} gave the first strongly polynomial-time algorithm using a scaling technique. More recently, V\'egh~\cite{Vegh14} gave another strongly-polynomial algorithm using a different scaling-based algorithm.

Fisher markets are a special case of more general Arrow-Debreu exchange markets. There are many convex programming formulations for linear exchange markets; see~\cite{DevanurGV16} for details. The first polynomial time algorithm was obtained by Jain~\cite{Jain07} based on the ellipsoid method. Ye~\cite{Ye07} obtained a polynomial time algorithm based on the interior-point method. The first combinatorial polynomial-time algorithm was developed by Duan and Mehlhorn~\cite{DuanM15}, and it was later improved in~\cite{DuanGM16}. More recently, Garg and V\'egh~\cite{GargV19} obtained the first strongly polynomial-time algorithm for this problem. 

Linear Fisher markets with either utility or earning limits were studied only recently~\cite{ColeDGJMVY17,BeiGHM16,BeiGHM17}, and equilibria in these models can be captured by extensions of Eisenberg-Gale and Shmyrev convex programs, respectively. In markets with utility limits, combinatorial polynomial time algorithms are obtained~\cite{BeiGHM16,Vegh14}, the set of equilibria forms a lattice, and equilibria with maximum or minimum prices can also be obtained efficiently~\cite{BeiGHM16}. In markets with earning limits, combinatorial polynomial-time algorithms are obtained in~\cite{ColeG18,BeiGHM17}. Any equilibrium can be refined to one with minimal or maximal prices in polynomial time~\cite{BeiGHM17}.  

\paragraph{Nash Social Welfare} The Nash social welfare is a classic objective for allocation of goods to agents. It was proposed by Nash~\cite{Nash50} for the bargaining problem as the unique objective that satisfies a collection of natural axioms. Since then it has received significant attention in the literature on social choice and fair division (see, e.g.\ \cite{CaragiannisKMPSW16, DarmannS15, RamezaniS10, FreemanZC17} for a subset of notable recent work, and the references therein).

For divisible items, the problem of maximizing the Nash social welfare is solved by competitive equilibria with equal incomes (CEEI)~\cite{Moulin03}. However, CEEI can provide significantly more value in terms of Nash social welfare than optimal solutions for indivisible items. To obtain an improved bound on the indivisible optimum, Cole and Gkatzelis~\cite{ColeG18} introduced and rounded \emph{spending-restricted equilibria}, i.e., equilibria in markets with an earning limit of 1 for every good. More generally, equilibria in linear markets with earning limits can be described by a convex program~\cite{ColeDGJMVY17} similar to the one by Shmyrev. 

For indivisible items and general non-negative valuations, the problem of maximizing the Nash social welfare is hard to approximate within any finite factor~\cite{NguyenNRR14}. For additive valuations, the problem is \classAPX-hard~\cite{Lee17}, and efficient 2-approximation algorithms based on market equilibrium~\cite{ColeG18,ColeDGJMVY17} and stable polynomials~\cite{AnariGSS17, AnariMGV18} exist. These algorithms have been extended to give a 2-approximation also in markets with multiple copies per item~\cite{BeiGHM17} and additive-separable concave valuations~\cite{AnariMGV18}. Barman et al.~\cite{BarmanKV18} introduced another technique based on limited envy and obtained a $1.45$-approximation for additive valuations. Very recently, Chaudhury et al.~\cite{ChaudhuryCGGHM18} generalized this result to obtain a $1.45$-approximation for a common generalization of both budget-additive and additive-separable concave valuations. 

\subsection{Outline}
The rest of the paper is structured as follows. We introduce notation and preliminaries in the following Section~\ref{sec:Prelim}. In Section~\ref{sec:exist} we discuss the existence of market equilibria under the money clearing condition. The FPTAS for perturbed markets is discussed in Section~\ref{sec:FPTAS}. The following sections contain our results on computing exact equilibria -- membership in \classPLS\ (Section~\ref{sec:PLS}), the polynomial-time algorithms for a constant number of buyers or goods (Section~\ref{sec:Constant}), and membership in PPAD (Section~\ref{sec:PPAD}). The rounding algorithm for maximizing the Nash social welfare and the analysis of its approximation factor are presented in Section~\ref{sec:NSWalgo}. In Section~\ref{sec:LB} we present the improved hardness bound for approximation of Nash social welfare with additive valuations. Finally, we conclude in Section~\ref{sec:Future} with a discussion of directions for future research.


\section{Preliminaries}\label{sec:Prelim}

\paragraph{Fisher Markets with Earning and Utility Limits} 
In such a market, there is a set $B$ of $n$ \emph{buyers} and a set $G$ of $m$ \emph{divisible goods}. Each good is owned by a separate seller and comes in unit supply. Each buyer $i \in B$ has a \emph{value} $u_{ij} \ge 0$ for a unit of good $j \in G$ and an \emph{endowment} $m_i \ge 0$ of money. Suppose buyer $i$ receives a bundle of goods $\vecx_i = (x_{ij})_{j \in G}$ with $x_{ij} \in [0,1]$, then the \emph{utility function} is budget-additive $u_i(\vecx_i) = \min\left(c_i, \sum_{j} u_{ij}x_{ij}\right)$, where $c_i > 0$ is the \emph{utility cap}.

The vector $\vecx = (\vecx_i)_{i \in B}$ with $\sum_{i \in B} x_{ij} = 1$ for every $j \in G$ denotes a \emph{(fractional) allocation} of goods to buyers. For an allocation, we call $i$ a \emph{capped buyer} if $u_i(\vecx_i) = c_i$. We also maintain a vector $\vecp = (p_1,\ldots,p_m)$ of \emph{prices} for the goods. Given price $p_j$ for good $j$, a buyer needs to pay $x_{ij} p_j$ when getting $x_{ij}$ allocation of good $j \in G$. Given a vector of prices $\vecp$, a \emph{demand bundle} $\vecx_i^*$ of buyer $i$ is a bundle of goods that maximizes the utility of buyer $i$ for its budget, i.e., $\vecx_i^* \in \arg\max_{\vecx_i} \left\{ u_i(\vecx_i) \mid \sum_{j} p_j x_{ij} \le m_i \right\}$. For price vector $\vecp$ and buyer $i$, we use $\lambda_i = \min_j p_j/u_{ij}$ and denote by $\alpha_i = 1/\lambda_i$ the \emph{maximum bang-per-buck (MBB)} ratio (where we assume $0/0 = 0$). Given prices $\vecp$ and allocation $\vecx$, the \emph{money flow} $f_{ij}$ from buyer $i$ to seller $j$ is given by $f_{ij} = p_j x_{ij}$. If price $p_j > 0$, then $x_{ij}$ uniquely determines $f_{ij}$ and vice versa. 

For the sellers, let $x_j = \sum_i x_{ij}$, then the seller utility is $u_j(x_j, p_j) = \min(d_j, p_j x_j)$ $= \min\left(d_j, \sum_{i} f_{ij}\right)$, where $d_j > 0$ is the \emph{earning} or \emph{income cap}. We call seller $j$ a \emph{capped seller} if $u_j(x_j, p_j) = d_j$. An \emph{optimal supply} $e_j^*$ allows seller $j$ to obtain the highest utility, i.e., $e^*_j \in \arg \max \left\{ u_j(e_j,p_j) \mid e_j \le 1 \right\}$.

We consider three natural properties for allocation and supply vectors:

\begin{enumerate}
\item An allocation $\vecx_i$ for buyer $i$ is called \emph{modest} if $\sum_j u_{ij}x_{ij} \le c_i$. By definition, for uncapped buyers every demand bundle is modest. For capped buyers, a modest bundle of goods $\vecx_i$ is such that $c_i = \sum_j u_{ij}x_{ij}$. 
\item A demand bundle $\vecx_i$ is called \emph{thrifty} or \emph{MBB} if it consists only of MBB goods: $x_{ij} > 0$ only if $u_{ij}/p_j = \alpha_i$. For uncapped buyers every demand bundle is MBB. 
\item A supply $e_j$ for seller $j$ is called \emph{modest} if $e_j = \min(1, d_j/p_j)$.
\end{enumerate}

Given a set of prices, a thrifty and modest demand bundle for buyer $i$ minimizes the amount of money required to obtain optimal utility. A modest supply for seller $j$ minimizes the amount of supply required to obtain optimal utility in equilibrium. Our interest lies in market equilibria that have thrifty and modest demands and modest supply. Note that they also emerge when earning and utility caps are not satiation points but \emph{limits in the form of hard constraints} on the utility in equilibrium (c.f.\ \cite{ColeDGJMVY17}).
\begin{definition}[Thrifty and Modest Equilibrium]
\label{def:equilibrium}
A \emph{thrifty and modest (market) equilibrium} is a pair $(\vecx,\vecp)$, where $\vecx$ is an allocation and $\vecp$ a vector of \emph{prices} such that the following conditions hold: (1) $\vecp \ge 0$ (prices are nonnegative), (2) $e_j$ is a modest supply for every $j \in G$, (3) $x_j \le e_j$ for every $j \in G$ (no overallocation), (4) $\vecx_i$ is a thrifty and modest demand bundle for every $i \in B$, and (5) Walras' law holds: $p_j (e_j - x_j) = 0$ for every $j \in G$. 
\end{definition}
Note that in equilibrium, if $x_j  < e_j$, then $p_j = 0$, due to Walras law. Moreover, we assume that all parameters of the market $u_{ij}$, $c_i$, $d_j$ and $m_i$ for all $i \in B$ and $j \in G$ are non-negative integers. Let $U = \max_{i \in B, j \in G} \{u_{ij}, m_i, c_i, d_j\}$ be the largest integer in the representation of the market.

Consider the following condition termed \emph{money clearing}: For each subset of buyers and the goods these buyers are interested in, there must be a feasible allocation of the buyer money that does not violate the earning caps. More formally, let $\tB \subseteq B$ be a set of buyers, and $N(\tB) = \{j\in G\ |\ u_{ij}>0 \text{ for some } i\in \tB\}$ be the set of goods such that there is at least one buyer in $\tB$ with positive utility for the good.
\begin{definition}[Money Clearing]
	A market is \emph{money clearing} if 
	\begin{equation}\label{eqn:nsc}
		\forall \tB \subseteq B,\; \sum_{i\in \tB}m_i \le \sum_{j\in N(\tB)} d_j \enspace.
	\end{equation}
\end{definition}
When there are only earning limits, money clearing is a precise characterization of markets that have thrifty and modest equilibria~\cite{BeiGHM17}. For markets with both limits, it is sufficient for existence (see Section~\ref{sec:exist}).

\paragraph{Perturbed Markets}
Our FPTAS in Section~\ref{sec:FPTAS} computes a thrifty and modest equilibrium in a \emph{perturbed market} $\tCM$.

\begin{definition}[Perturbed Utility, Perturbed Market]
For a market $\CM$ and a parameter $\eps > 0$, the \emph{perturbed utility} of buyer $i$ is given by $\tu_i(\vecx_i) = \sum_j \tu_{ij} x_{ij}$, where $\tu_{ij} \in \{0,(1+\eps)^{k} \; | \text{ integer } k \ge 0\}$ such that
\begin{equation}\label{eqn:approxU}
\tu_{ij}/(1+\eps) \le u_{ij} \le \tu_{ij}, \;\; \forall i\in B, j\in G.
\end{equation}
The \emph{perturbed market} $\tCM$ is exactly the market $\CM$, in which every buyer $i \in B$ has perturbed utilities $\tu_i$.
\end{definition}
In Section~\ref{sec:FPTAS} we observe that an exact equilibrium in $\tCM$ represents an $\eps$-approximate equilibrium for the unperturbed market $\CM$. 

\paragraph{Nash Social Welfare} 
There is a set $B$ of $n$ \emph{agents} and a set $G$ of $m$ \emph{indivisible items}, where we assume $m \ge n$. We allocate the items to the agents, and we represent an allocation $S = (S_1,\ldots,S_n)$ using a characteristic vector $\x^S$ with $x^S_{ij} = 1$ iff $j \in S_i$ and 0 otherwise. Agent $i \in B$ has a \emph{value} $v_{ij} \ge 0$ for item $j$ and a global \emph{utility cap} $c_i > 0$. The \emph{budget-additive valuation} of agent $i$ for an allocation $S$ of items is $v_i(\x_i^S) = \min\left( c_i, \sum_{j \in G} v_{ij}x_{ij}^S \right)$. The goal is to find an allocation that approximates the optimal Nash social welfare, i.e., the optimal geometric mean of valuations
 \[\max_{S} \left(\prod_{i \in B} v_i(\x_i^S)‚ \right)^{1/n}\enspace.\]
Our approximation algorithm in Section~\ref{sec:NSW} relies on rounding an equilibrium for a linear Fisher market with earning and utility limits. Our rounding algorithm loses a constant factor in the Nash social welfare. More precisely, we round an exact equilibrium of the perturbed market $\tCM$. The fact that this equilibrium satisfies the properties in Definition~\ref{def:equilibrium} with respect to perturbed utilities deteriorates the approximation factor only by a small constant (see Section~\ref{sec:perturbNSW}).


\newcommand{\q}{{\mathbf q}}
\newcommand{\g}{{\mathbf g}}
\newcommand{\s}{{\mathbf s}}

\newcommand{\nfrac}{\nicefrac}

\newcommand{\bDelta}{{\bm{\delta}}}
\newcommand{\bBeta}{{\bm{\beta}}}
\renewcommand{\l}{{\bm{\lambda}}}
\newcommand{\li}{{\lambda_i}}
\renewcommand{\L}{{\Lambda}}

\renewcommand{\P}{{\p^{min}}}

\newcommand{\J}[1]{{\color{blue} #1}}
\newcommand{\defeq}{\stackrel{\textup{def}}{=}}

\section{Computing Equilibria}

\subsection{Existence and Structure of Equilibria}
\label{sec:exist}

In this section, we briefly discuss existence and structure of thrifty and modest equilibria in markets with utility and earning limits. The set of equilibria in these markets has interesting and non-trivial structure. For markets with utility limits, an equilibrium always exists~\cite{BeiGHM16}. For markets with earning limits, an equilibrium may not exist, because uncapped buyers always spend all their money. In these markets, the money-clearing condition is necessary and sufficient for the existence of a thrifty and modest equilibrium~\cite{BeiGHM17} (see also~\cite{ColeDGJMVY17} for the case that $u_{ij} > 0$ for all $i \in B$, $j \in G$). 

We observe that in a market $\CM$ with both limits, money clearing is sufficient but not necessary for the existence of a thrifty and modest equilibrium. Our FPTAS below gives an $\eps$-approximate equilibrium in money-clearing markets, for arbitrarily small $\eps$. Since market parameters are finite integers, for sufficiently small $\eps$ this implies existence of an exact equilibrium. 

This is interesting since the structure of equilibria in such markets can be quite complex. For example, in money-clearing markets $\CM$ there can be no convex program describing thrifty and modest equilibria. This holds even if we restrict to the ones that are Pareto-optimal with respect to the set of all thrifty and modest equilibria. Equilibria for the corresponding markets \emph{without} caps, or with \emph{either} earning \emph{or} utility caps might not remain equilibria in the market with \emph{both} sets of caps. Hence, existence of a thrifty and modest equilibrium in money-clearing markets $\CM$ follows neither from a convex program nor by a direct application of existing algorithms for markets with only one set of either utility or earning caps. The following proposition summarizes our observations.

\begin{proposition}
\label{prop:structure}
There are markets $\CM$ with utility and earning limits such that the following holds:
\begin{enumerate}
\setlength{\tabcolsep}{0pt}
\setlength{\itemsep}{0pt}
\setlength{\parsep}{0pt}
\item $\CM$ is not money-clearing and has a thrifty and modest equilibrium.
\item $\CM$ is money-clearing, and the set of thrifty and modest equilibria is not convex. Among these equilibria, there are multiple Pareto-optimal equilibria, and their set is also not convex.
\item For a money-clearing market $\CM$ and the three related markets -- (1) with only utility caps, (2) with only earning caps, (3) without any caps -- the sets of equilibria are mutually disjoint.
\end{enumerate}
\end{proposition}

\begin{proof}
We provide an example market for each of the three properties. \\

\noindent \textbf{Property 1:} Consider a linear market with one buyer and one good. The buyer has $m_1 = 2$, utility $u_{11} = 2$, and utility cap $c_1 = 1$. The good has earning cap $d_1 = 1$. The unique thrifty and modest equilibrium has price $p_1 = 2$ and allocation $x_{11} = 1/2$. Both seller and buyer exactly reach their cap. The active budget $m_i^a = 1$ equals the earning cap. Conversely, due to price 2, the supply is 1/2, for which the achieved utility equals the utility cap. Note that the money clearing condition~\eqref{eqn:nsc} is violated. \\

\noindent \textbf{Property 2:}
Consider the following example. There are two buyers and two goods. The buyer budgets are $m_1 = 2$ and $m_2 = 32$. The utility caps are $c_1 = \infty$, $c_2 = 32$, the earning caps are $d_1 = 8$, $d_2 = 26$. The linear utilities are given by the parameters $u_{11} = u_{22} = 32$, $u_{12} = 128$, and $u_{21} = 2$.

If we ignore all caps, the unique equilibrium has prices $(2,32)$ and buyer utilities $(32,32)$. If we ignore the utility caps and consider only earning caps, the equilibrium prices are $(8y,128y)$ and buyer utilities are $(\nicefrac{8}{y}, \nicefrac{8}{y})$, for $y \ge 1$. If we ignore the earning caps and consider only utility caps, the equilibrium prices are $(2,x)$ and buyer utilities are $(32,32)$, for $x \in [8,32]$. 

With all caps, the equilibria form two disjoint convex sets: either prices $(2,x)$ and buyer utilities $(32,32)$, for $x \in [8,26]$; or prices $(8y,128y)$ and buyer utilities $(\nicefrac{8}{y}, \nicefrac{8}{y})$, for $y \ge 1$. Note that $(2,x)$ for $x \in (26,32]$ are not equilibrium prices, since this would violate the earning cap of seller 2.

Observe that there are exactly two Pareto-optimal equilibria: prices $(1,6)$ (which also represents income for the sellers) and buyer utilities $(1,1)$; and prices $(5,50)$ (with income $(3,6)$ for the sellers) and buyer utilities $(\nicefrac{1}{5},\nicefrac{1}{5})$. The first equilibrium is strictly better for both buyers, the second one strictly better for seller 1.\\

\noindent \textbf{Property 3:}
Consider the following market with 2 buyers and 2 goods. The buyer budgets are $m_1 = 100$ and $m_2 = 11$. The utility caps are $c_1 = 0.9$, $c_2 = \infty$. The earning caps are $d_1 = 9$, $d_2 = \infty$. The utilities are $u_{11} = u_{22} = u_{12} = u_{21} = 1$.

If we ignore all caps, the unique equilibrium prices are $(55.5, 55.5)$. If we ignore the buyer caps and consider only seller caps, the unique equilibrium prices are $(102, 102)$. If we ignore the seller caps and consider only buyer caps, the unique equilibrium prices are $(10, 10)$. For both buyer and seller caps, the unique equilibrium prices are $(20, 20)$.
\end{proof}

\subsection{Computing Equilibria in Perturbed Markets}
\label{sec:FPTAS}

In this section, we describe and analyze Algorithm~\ref{alg:FPTAS}, an FPTAS for computing an approximate equilibrium in money-clearing markets $\CM$. The input parameters for such a market are $u_{ij}, m_i, c_i, d_j, \forall i\in B, j\in G$, where $u_{ij}$ is the utility derived by buyer $i$ for a unit amount of good $j$, $m_i$ is the budget of buyer $i$, $c_i$ is the utility cap of buyer $i$, and $d_j$ is the earning cap of seller $j$. For any $\eps > 0$, Algorithm~\ref{alg:FPTAS} computes an exact equilibrium in a perturbed market $\tCM$, where we increase every non-zero parameter $u_{ij}$ to the next-larger power of $(1+\eps)$. 

\begin{algorithm}[t]
\caption{\label{alg:FPTAS} FPTAS for $\CM$ with Earning and Utility Caps}
\DontPrintSemicolon
\SetKwInOut{Input}{Input}\SetKwInOut{Output}{Output}
\Input{Market $\CM$ given by budgets $m_i$, utility caps $c_i$, earning caps $d_j$, utilities $u_{ij}, \forall i \in B, j \in G$, approximation parameter $\eps$;}
\Output{Equilibrium $(\x,\p)$ of the perturbed market $\tCM$}
Construct $\tCM$ by increasing each non-zero $u_{ij}$ to the next-larger power of $(1+\epsilon)$, set $\tU \ot \max_{ij} \tu_{ij}$, and run the rest of the algorithm on $\tCM$\;
$(\f, \p) \ot$ equilibrium of $\tCM$ when ignoring all utility caps \; 
$Z \ot \{i\in B \mid s(i) = 0\}$ \tcp*{set of zero surplus buyers}
\While{$Z \neq B$}{
	$k \ot$ a buyer in $B\setminus Z$\tcp*{$s(k) > 0$}
	\While{$(s(k) > 0)$ and $(\min_{j\in G:p_j > 0} p_j > 1/n\tU^n)$}{
		$\tB \ot \{k\} \cup \{i\in B \mid i \text{ can reach $k$ in the MBB residual graph}\}$\;
		$\tG \ot \{j\in G \mid j \text{ can reach $k$ in the MBB residual graph}\}$\;
		$\p'\ot \p$ and $\gamma\ot 1$\; 
        Define $p_j \ot \gamma \cdot p_j, \forall j\in \tG$, and adjust active prices and budgets accordingly as a function of $\gamma$\;
		Decrease $\gamma$ continuously down from 1 until one of these events occurs:\;
		\ \ \ \ \ {\bf Event 1:} A new MBB edge appears\;
		\ \ \ \ \ {\bf Event 2:} $\gamma =$ MinFactor$(\p', \f, \tB, \tG, Z)$ \tcp*{Algorithm \ref{alg:minFactor}} 
		$\f \ot$ FeasibleFlow$(\p, Z)$\tcp*{Algorithm \ref{alg:FeasibleFlow}}
	}
	\If{$\min_{j:p_j > 0} p_j \le 1/n\tU^n$}{
		Choose any good $\ell \in \arg\min \{ p_j \mid p_j > 0 \}$\;
		$\tG \ot \{\ell\} \cup \{j\in G \mid j \text{ is connected to $\ell$ in the MBB graph } \}$\;
		$\tB \ot \{i\in B \mid \tu_{ij} > 0 \text{ for some } j\in \tG\}$\;
          Assign $(\x_i)_{i \in \tB}$ according to $\f$\;
		$s(i) \ot 0, \forall i\in \tB$ and $p_j \ot 0, \forall j \in \tG$\;
	}
	$Z \ot Z \cup \{i\in B\ |\ s(i) = 0\}$\;\label{line:Zinc}
}
Assign $\x_i$ according to $\f$ for all buyers $i \in B$ that have not been assigned yet.\;
\Return $(\x,\p)$
\end{algorithm}

\paragraph{Additional Concepts}
Our algorithm steers prices and money flow towards equilibrium by monitoring the surplus of buyers and sellers. Note that a buyer $i$ is capped if $m_i \alpha_i \ge c_i$. 

\begin{definition}[Active Budget, Active Price, Surplus]
Given prices $\p$ and flow $\f$, the \emph{active budget} of buyer $i$ is $m_i^a = \min(m_i, c_i/\alpha_i)$, the \emph{active supply} of seller $j$ is $e_j^a = \min(1, d_j/p_j)$, and the \emph{active price} is $p_j^a = p_j e_j^a = \min(p_j,d_j)$. The \emph{surplus of buyer $i$} is $s(i) = \sum_{j \in \G} f_{ij} - m_i^a$, and the \emph{surplus of good $j$} is $s(j) = p_j^a - \sum_{i\in \B} f_{ij}$.
\end{definition}

Several graphs connected to the MBB ratio are useful here. As argued in~\cite{Orlin10,DuanGM16}, we can assume w.l.o.g. that the MBB graph is \emph{non-degenerate}, i.e., it is a forest. 

\begin{definition}[MBB edge, MBB graph, MBB residual graph] 	
\label{def:MBBgraphs}
\label{def:MBBedge}
	Given prices $\p$, an undirected pair $\{i,j\}$ is an \emph{MBB edge} if $i \in B$, $j \in G$, and $u_{ij}/p_j = \alpha_i$. The \emph{MBB graph} $\CG(\p) = (B \cup G, E)$ is an undirected graph that contains exactly the MBB edges. Given prices $\p$ and money flow $\f$, the \emph{MBB residual graph} $\CG_r(\f,\p) = (B \cup G, A)$ is a directed graph with the following arcs: If $\{i,j\}$ is MBB, then $(i,j)$ is an arc in $A$; if $\{i,j\}$ is MBB and $f_{ij} > 0$, then $(j,i)$ is an arc in $A$.
\end{definition}

Let us also define a \emph{reverse flow network} $N^-(\p,Z)$ by adding a sink $t$ to the MBB graph. The network has nodes $G \cup B \cup \{t\}$, edges $(i, t)$ for $i\in B\setminus Z$, and the reverse MBB edges $(j,i)$ if $(i,j)$ is an MBB edge. All edges have infinite capacity. The supply at node $j \in G$ is $p_j^a$, demand at node $i \in B$ is $m_i^a$, and demand at node $t$ is $\sum_j p_j^a - \sum_i m_i^a$. The flow in the network corresponds to money. Given a money flow $\f$ in the network $N^-(\p, Z)$, the surplus of buyer $i \in B\setminus Z$ corresponds to flow on $(i,t)$
\[s(i) = \sum_{j \in G} f_{ij} - m_i^a = f_{it}\enspace.\] 
Buyers in $Z$ do not have edges to the sink. Hence, their surplus is fixed to 0 at every feasible flow. 

\begin{algorithm}[t]
\caption{MinFactor \label{alg:minFactor}}
\DontPrintSemicolon
\SetKwInOut{Input}{Input}\SetKwInOut{Output}{Output}
\Input{Prices $\p$, flow $\f$, set of buyers $\tB$, set of goods $\tG$, set of zero-surplus buyers $Z$}
\Output{Minimum price decrease consistent with the input configuration}
$E \ot$ Set of MBB edges at prices $\p$ between $\tB$ and $\tG$\;
$G_c \ot$ Set of goods from $\tG$ that are capped at $(\f,\p)$\;
$B_c \ot$ Set of buyers from $\tB$ that are capped at $(\f,\p)$\;
$\li \ot \min_{k\in G}p_k/u_{ik}, \forall i\in \tB$\;
Set up the following LP in flow variables $\g$ and $\gamma$:\;
\vspace{0.1cm} \nonl \ \ \ \ 
\begin{tabular}{|ll|}
\hline
$\min \gamma$ & \\
$\sum_{i\in\tB} g_{ij} = d_j$, & $\forall j\in G_c$\\
$\sum_{i\in\tB} g_{ij} = \gamma p_j$, & $\forall j\in \tG\setminus G_c$\\
$\sum_{j\in\tG} g_{ij} = \gamma c_i\li$, & $\forall i\in B_c \cap Z$\\
$\sum_{j\in\tG} g_{ij} \ge \gamma c_i\li$, & $\forall i\in B_c \setminus Z$\\
$\sum_{j\in\tG} g_{ij} = m_i$, & $\forall i\in (\tB\setminus B_c) \cap Z$\\
$\sum_{j\in\tG} g_{ij} \ge m_i$, & $\forall i\in (\tB\setminus B_c) \setminus Z$\\
$g_{ij} = 0$, & $\forall (i,j) \not\in E$\\
$g_{ij} \ge 0$, & $\forall i\in \tB, j\in \tG$\\
\hline
\end{tabular} \; 
\vspace{0.1cm}
\Return{Optimal solution $\gamma$ of above LP}
\end{algorithm}

\begin{algorithm}[t]
\caption{FeasibleFlow \label{alg:FeasibleFlow}}
\DontPrintSemicolon
\SetKwInOut{Input}{Input}\SetKwInOut{Output}{Output}
\Input{Perturbed market $\tCM$, prices $\p$, and set of zero-surplus buyers $Z$}
\Output{Feasible flow consistent with the input configuration}
$E \ot$ Set of MBB edges at prices $\p$\; 
$\li \ot \min_{k\in G}p_k/u_{ik}, \forall i\in B$\;
$B_c \ot$ Set of capped buyers at $\p$\;
$G_c \ot$ Set of capped goods at $\p$\;
Set up the following feasibility LP in flow variables $\f$:\;
\vspace{0.1cm}\nonl \ \ \ \ 
\begin{tabular}{|ll|}
\hline
$\sum_{i\in\tB} f_{ij} = d_j$, & $\forall j\in G_c$\\
$\sum_{i\in\tB} f_{ij} = p_j$, & $\forall j\in G\setminus G_c$\\
$\sum_{j\in\tG} f_{ij} = c_i\li$, & $\forall i\in B_c \cap Z$\\
$\sum_{j\in\tG} f_{ij} \ge c_i\li$, & $\forall i\in B_c \setminus Z$\\
$\sum_{j\in\tG} f_{ij} = m_i$, & $\forall i\in (\B\setminus B_c) \cap Z$\\
$\sum_{j\in\tG} f_{ij} \ge m_i$, & $\forall i\in (\B\setminus B_c) \setminus Z$\\
$f_{ij} = 0$, & $\forall (i,j) \not\in E$\\
$f_{ij} \ge 0$, & $\forall i\in \B, j\in \G$\\
\hline
\end{tabular} \\
\vspace{0.1cm}
\Return{Feasible solution $f$ of above LP}
\end{algorithm}

\paragraph{Algorithm and Analysis}
Algorithm \ref{alg:FPTAS} computes an exact equilibrium of $\tCM$. For convenience, it maintains a money flow $\f$. For goods with non-zero price, $\f$ is equivalent to an allocation $\x$. When the algorithm encounters a set of goods with price 0, the buyers interested in these goods must be capped, and the algorithm determines a suitable allocation for them by solving a system of linear equations.

The algorithm first calls a subroutine to compute a market equilibrium ignoring the utility caps of the buyers. Such an equilibrium exists because the market is money-clearing, can be computed in polynomial time~\cite{ColeDGJMVY17,BeiGHM17}, and consists of a pair $(\f, \p)$ of flow and prices such that the outflow of every good $j$ is $p_j^a$ and the inflow of every buyer $i$ is $m_i$. Given this equilibrium, the algorithm then initializes $Z$ to the set of buyers with surplus zero in $(\f, \p)$. 

The following {\bf Invariants} are maintained during the run of Algorithm~\ref{alg:FPTAS}:
\begin{itemize}
\setlength{\tabcolsep}{0cm}
\setlength{\itemsep}{0cm}
\setlength{\parsep}{0cm}
\item no price ever increases.
\item if $s(i) = 0$ for a buyer $i$, it remains 0. $Z$ is monotonically increasing. 
\item $N^-(\p,Z)$ allows a feasible flow, i.e., $s(i) \ge 0$ for every buyer $i \in B$ and $s(j) = 0$ for every good $j \in G$. 
\end{itemize}
More formally, the algorithm uses a descending-price approach. There is always a flow in $N^-(\p, Z)$ with outflow of a good $j\in G$ equal to $p_j^a$, in-flow into buyer $i\in B\cap Z$ equal to $m_i^a$, and in-flow into buyer $i\in B\setminus Z$ at least $m_i^a$. Descending prices imply that if a good (buyer) becomes uncapped (capped), it remains uncapped (capped). 

The algorithm ends when $Z=B$, i.e., all buyers have surplus zero, and hence $(\f,\p)$ is an equilibrium of $\tCM$. In the body of the outer while-loop, we first pick a buyer $k$ whose surplus is positive. The inner while loop ends when either the surplus of $k$ becomes zero or the minimum positive price of a good, say $\ell$, is at most $1/n\tU^n$, where $\tU$ is the maximum parameter value of the perturbed utilities. In the former case, the size of $Z$ increases (in line~\ref{line:Zinc}). In the latter case, we obtain a set $\tG$ of goods connected to $\ell$ through MBB edges and a set $\tB$ of buyers who have non-zero utility for some good in $\tG$. Since the price of each good in $\tG$ is so low and their surplus is zero, each buyer in $\tB$ must be capped. Hence we fix the allocation of buyers in $\tB$ according to the current money flow $\f$, and set the prices of all goods in $\tG$ and surplus of all buyers in $\tB$ to zero. Since the algorithm maintains goods with price 0 and buyers with surplus 0, the inner-while loop is executed at most $m + n$ times. 

In the body of inner while-loop, we construct the set $\tB$ of buyers and $\tG$ of goods that can reach buyer $k$ in the MBB residual graph (see Definition \ref{def:MBBgraphs}). We then continuously decrease the prices of all goods in $\tG$ by a common factor $\gamma$, starting from $\gamma=1$. This may destroy MBB edges connecting buyers in $\tB$ with goods in $G\setminus\tG$. However, by definition of $\tG$ there is no flow on such edges. For uncapped goods in $\tG$ (capped buyers in $\tB$), this decreases the active price (budget) by a factor of $\gamma$. We stop if one of the two events happens: 
$(1)$ a new MBB edge appears, and $(2)$ $\gamma$ is equal to the minimum factor possible that allows a feasible flow with the current MBB edges, i.e., in-flow into a good $j\in \tG$ is equal to $p_j^a$, out-flow of a buyer in $\tB\cap Z$ is equal to $m_i^a$, and out-flow of a buyer in $\tB\setminus Z$ is at least $m_i^a$. While the value of $\gamma$ for event (1) results from ratios of $\tu_{ij}$, the value of $\gamma$ for event (2) is found by Algorithm \ref{alg:minFactor} based on a linear program (LP). Observe that the flow $\f$ and $\gamma=1$ are a feasible initial solution for the LP. 

After the event happened, we update to a new feasible flow $\f$ using Algorithm \ref{alg:FeasibleFlow}. For prices $\p$ and the set $Z$ of zero-surplus buyers, the in-flow into a good $j\in G$ must be equal to $p_j^a$, out-flow of a buyer in $Z$ must be equal to $m_i^a$, and out-flow of a buyer in $B\setminus Z$ must be at least $m_i^a$. Algorithm \ref{alg:FeasibleFlow} sets up a feasibility LP to find such a feasible flow. Observe that this feasibility set is non-empty due to Event $2$. 

The following lemma is straightforward, we omit the proof. 

\begin{lemma}\label{lem:inv}
	The Invariants hold during the run of Algorithm \ref{alg:FPTAS}. 
\end{lemma}

Next we bound the running time of Algorithm \ref{alg:FPTAS}. Event 1 provides a new MBB edge between a buyer in $B\setminus \tB$ and a good in $\tG$. Event 2 restricts the price decrease in $\gamma$ such that the Invariants are maintained. The event happens only if (1) at the value of $\gamma$ there is a subset of buyers $S\subseteq \tB$ such that $\sum_{i\in S} m_i^a = \sum_{j\in \Gamma(S)} p_j^a$, where $\Gamma(S)$ is the set of goods to which buyers in $S$ have MBB edges, and (2) further decrease of prices would make the total active budget of buyers in $S$ more than the total active prices of $\Gamma(S)$. This condition would violate the invariant that $N^-(\p, Z)$ has a feasible flow where the surplus of each good is zero. 

If the subset $S$ is equal to $\tB$ or $S$ contains buyer $k$, then the surplus of $k$ in every feasible flow is zero at such a minimum $\gamma$, and hence the inner-while loop ends. Otherwise, the MBB edges between buyers in $B\setminus S$ and goods in $\Gamma(S)$ will become non-MBB in the next iteration. So in each event of the inner-while loop, either a new MBB edge evolves or an existing MBB edge vanishes. Next, we show that for a given buyer $k$, the total number of iterations of the inner-while loop is polynomially bounded. For this, we first show that price of a good strictly decreases during each iteration of inner-while loop. 

\begin{lemma}\label{lem:price}
 	In each iteration of inner-while loop, the MBB ratio of buyer $k$ strictly increases. 
\end{lemma}

\begin{proof}
Each iteration of the inner while-loop ends with one of the two events. Clearly, Event 1 can occur only when the prices of goods in $\tG$ strictly decrease, and this implies that the MBB of buyer $k$ strictly increases. In case of Event 2, as argued above, there is a subset $S\subseteq \tB$ of buyers such that $\sum_{i\in S} m_i^a = \sum_{j\in \Gamma(S)} p_j^a$, where $\Gamma(S)$ is the set of goods to which $S$ have MBB edges. 

If $k\in S$, then $s(k) = 0$ in this iteration. This implies that $\sum_{i\in S} m_i^a < \sum_{j\in \Gamma(S)}p_j^a$ at the beginning of this iteration, and since equality emerges, prices must have strictly decreased and the MBB of $k$ strictly increased. 

If $k\not\in S$, then $S \neq \tB$ and flow on all MBB edges from $\tB\setminus S$ to $\Gamma(S)$ has become zero. Note that there is at least one such edge due to the construction of $\tB$ and $\tG$. Using the fact that there was a non-zero flow on these edges and $\sum_{i\in S} m_i^a < \sum_{j\in \Gamma(S)}p_j^a$ at the beginning of this iteration, we conclude that prices of goods must have strictly decreased and the MBB of $k$ strictly increased.
\end{proof}

Next we show that the price of a good substantially decreases after a certain number of iterations. For this, we partition the iterations into \emph{phases}, where every phase has $n^2$ iterations of the inner while-loop. 

\begin{lemma}\label{lem:key}
Let $\p$ and $\p'$ be the prices at the beginning and end of a phase, respectively. Then $p'_j \le p_j, \forall j\in G$, and there exists a good $\ell$ such that $p'_{\ell} \le p_{\ell}/(1+\eps)$. 
\end{lemma}

\begin{proof}
Due to Lemma \ref{lem:inv}, we have $p'_j \le p_j, \forall j\in G$. For the second part, note that $\tB$ always contains buyer $k$ during an entire run of inner while-loop. Since prices monotonically decrease, the MBB $\ak$ of buyer $k$ monotonically increases. Further, if there is a MBB path from buyer $k$ to a good $j$, then we have, for some $(i_1, j_1),\ldots,(i_a,j_a),(i'_1,j'_1),\ldots,(i'_b,j'_b)$ and an integer $c$
\[\ak p_j = \frac{\prod \tu_{i_1j_1}\dots\tu_{i_aj_a}}{\prod \tu_{i'_1j'_1}\dots\tu_{i'_bj'_b}} = (1+\eps)^c \enspace.\]
In each iteration, either a new MBB edge evolves or an existing MBB edge vanishes. When a new MBB edge evolves, a new MBB path from buyer $k$ to a good $j$ gets established. When an existing MBB edge vanishes, then an old MBB path from $k$ to a good $j$ gets destroyed. Further, if there is an MBB path from a good $j$ to buyer $k$, then price of good $j$ monotonically decreases. If there is no MBB path from a good $j$ to buyer $k$, then price of good $j$ does not decrease. After $O(n)$ events, 
there has to be a good $j$ such that initially there is an MBB path from $k$ to $j$, then no MBB path between them for some iterations, then again an MBB path between them. Let $p_j$ be the price of good $j$ at the time when there is no path between $k$ and $j$, and let $\ak$ and $\ak'$ be the MBB for buyer $k$ at the time the MBB path between $j$ and $k$ was broken and when it was later again established, respectively. Since $p_j$ does not change unless there is a path between $k$ and $j$, we have
\[ \ak p_j = (1+\eps)^{c_1} \text{ and } \ak'p_j = (1+\eps)^{c_2}, \text{ for some integers } c_1 \text{ and } c_2.\]
Since $\ak' > \ak$ due to Lemma \ref{lem:price}, we have $\ak' \ge \ak (1+\eps)$. Let good $l$ give the MBB to buyer $k$ at $\ak'$, and let $p_l$ and $p_l'$ be the prices of good $l$ when the MBB path between $j$ and $k$ was broken and when it was later established. This implies  
\[u_{il}/p'_l = \ak' \ge \ak(1+\eps) \ge (1+\eps)u_{il}/p_l,\]
and $p'_l \le p_l/ (1+\eps)$. 
\end{proof}

\begin{lemma}\label{lem:niter}
The number of iterations of inner while-loop of Algorithm~\ref{alg:FPTAS} is $O(n^3\log_{1+\eps}(n\tU^n\sum_i m_i))$. 
\end{lemma}

\begin{proof}
From Lemma \ref{lem:key}, in each phase the price of a good decreases by a factor of $(1+\eps)$. The number of iterations in a phase is $O(n^2)$. The starting price is at most $\sum_i m_i$. If a price becomes at most $1/n\tU^n$, the inner while-loop ends for a particular buyer $k$. Hence, the number of phases is at most $n\log_{1+\eps}{n\tU^n\sum_i m_i}$, and the number of iterations of the inner while-loop is at most $O(n^3\log_{1+\eps}{n\tU^n\sum_i m_i})$. 
\end{proof}

\begin{theorem}
For every $\eps > 0$, Algorithm \ref{alg:FPTAS} computes a thrifty and modest equilibrium in the perturbed market $\tCM$ in time polynomial in $n$, $U$ and $1/\eps$.
\end{theorem}

\begin{proof}
From Lemma~\ref{lem:inv}, all invariants are maintained throughout the algorithm. Hence, the surplus of each good is $0$, the surplus of each buyer is non-negative, and prices decrease monotonically. The algorithm ends when surplus of all buyers is zero. During the algorithm, when the price of a good, say $\ell$, becomes at most $1/n\tU^n$, where $\tU$ is the largest perturbed utility parameter, then the price of all the goods connected to $\ell$ by MBB edges is at most $1/n$. Since the minimum budget of a buyer is at least $1$, all buyers buying these goods have to be capped. That implies that there is an equilibrium where prices of these goods are zero. 

Lemma~\ref{lem:niter} shows that there are at most $O(n^3\log_{1+\eps}{n\tU^n\sum_i m_i})$ iterations, which can be upper bounded by $O(\nicefrac{n^4}{\eps} \log (nU))$. Each iteration can be implemented in polynomial time. 
\end{proof}

\paragraph{Approximate Equilibrium}
Our algorithm computes an exact equilibrium in $\tCM$ in polynomial time. We show that such an exact equilibrium of $\tCM$ represents an $\eps$-approximate equilibrium of $\CM$, thereby obtaining an FPTAS for the problem. Based on $\eps$, let us define the precise notion of $\eps$-approximate market equilibrium, which is based on a notion of $\eps$-approximate demand bundle.

\begin{definition}[Approximate Demand]
For a vector $\vecp$ of prices, consider a demand bundle $\vecx_i^*$ for buyer $i$. An allocation $\vecx_i$ for buyer $i$ is called an \emph{$\eps$-approximate (thrifty and modest) demand bundle} if (1) $\sum_j u_{ij} x_{ij} \le c_i$, (2) $\sum_j x_{ij} p_j \le m_i^a$, and (3) $u_i(\vecx_i) \ge (1-\eps) u_i(\vecx_i^*)$.
\end{definition}

An $\eps$-approximate (thrifty and modest) equilibrium differs from an exact equilibrium only by a relaxation of condition (4) to \emph{$\eps$-approximate} demand (c.f.\ Definition~\ref{def:equilibrium})

\begin{definition}[Approximate Equilibrium]
An \emph{$\eps$-approximate (thrifty and modest) equilibrium} is a pair $(\vecx,\vecp)$, where $\vecx$ is an allocation and $\vecp$ a vector of \emph{prices} such that conditions (1)-(3), (5) from Definition~\ref{def:equilibrium} hold, and (4) $\vecx_i$ is an $\eps$-approximate demand bundle for every $i \in B$.
\end{definition}
Note that our definition is rather demanding, since there are many further relaxations (e.g., we require exact market clearing, modest supplies, exact earning and utility caps, etc), some of which are found in other notions of approximate equilibrium in the literature.

\begin{lemma}
\label{lem:red}
An exact equilibrium $(\x,\p)$ of $\tCM$ is an $\eps$-approximate equilibrium of $\CM$. 
\end{lemma}

\begin{proof}
Let $\ai$ and $\tai$ be the MBB of buyer $i$ at prices $\p$ w.r.t.\ utility $u_i$ and perturbed utility $\tu_i$, respectively. Formally, $\ai = \max_{k\in G} \nfrac{u_{ik}}{p_k}$ and $\tai = \max_{k\in G} \nfrac{\tu_{ik}}{p_k}$. At prices $\p$, let $u_i^*$ and $\tu_i^*$ be the maximum utility buyer $i$ can obtain in $\CM$ and $\tCM$, respectively. Clearly $u_i^* = \min\{c_i, m_i\ai\}$, and $\tu_i^* = \min\{c_i, m_i\tai\}$.

Since $(\x, \p)$ is an exact equilibrium of $\tCM$, the MBB condition implies that $x_{ij} > 0$ only if $\tu_{ij}/p_j = \max_{k\in G} \tu_{ik}/p_k$. Further, using~\eqref{eqn:approxU} we get $\tai (1+\eps) > \ai, \forall i$. This implies that $\tu_i^* > u_i^*/(1+\eps) \ge u_i^* (1-\eps)$. Further, since $\tu_{ij} \ge u_{ij}$, we have $\sum_j x_{ij}p_j = m_i^a, \forall i$. In addition, since $(\x,\p)$ is an exact equilibrium for $\tCM$, we obtain $\sum_i x_{ij} = \min\{1, p_j^a/p_j\}, \forall \{j\in G\ |\ p_j > 0\}$ and $\sum_i x_{ij} \le 1, \forall \{j\in G\ |\ p_j = 0\}$. This proves the claim. 
\end{proof}

\begin{corollary}
 Algorithm \ref{alg:FPTAS} is an FPTAS for computing an $\eps$-approximate equilibrium for money-clearing markets with earning and utility limits.
\end{corollary}

\subsection{Membership in PLS}
\label{sec:PLS}

In this section we show that the problem of computing an exact equilibrium in a money-clearing market $\CM$ is in the class \classPLS. We first design Algorithm \ref{alg:PLS}, a finite-time descending-price algorithm. It again relies on the \emph{reverse flow network} $N^-(\p) = N^-(\p, \emptyset)$ defined in the previous section where $Z$ is an empty set.
The algorithm starts by computing a market equilibrium ignoring the utility caps of the buyers. This equilibrium exists since the market is money-clearing. It is a pair $(\f,\p)$ of flow and prices for which the outflow of good $j$ is equal to $p_j^a$ and the inflow into buyer node $i$ is $m_i$. 
\medskip

\noindent{We will maintain the following \textbf{Invariants} during the while-loop in Algorithm~\ref{alg:PLS}:\\
$-$ No price ever increases.\\
$-$ $N^-(\p)$ allows a feasible flow.
\medskip

In other words, our algorithm is descending-price, and there is always a flow in $N^-(\p)$ with out-flow of good $j \in G$ equal to $p_j^a$ and in-flow into buyer $i \in B$ at least $m_i^a$. The first invariant implies that once a good becomes uncapped, it remains uncapped, and once a buyer becomes capped, it remains capped.

\begin{algorithm}[t]
\caption{\label{alg:PLS} Finite-Time Algorithm for Money-Clearing Markets 
}
\DontPrintSemicolon
\SetKwInOut{Input}{Input}\SetKwInOut{Output}{Output}
\Input{Market $\CM$ given by budgets $m_i$, utility caps $c_i$, earning caps $d_j$, utilities $u_{ij}, \forall i \in B, j \in G$;}
\Output{Equilibrium prices $\p$ and allocation $\x$}
$(\f, \p) \ot$ equilibrium of $\CM$ when ignoring all utility caps \; 
\While{$\sum_i s(i) > 0$}{
	$\f \ot$ balanced flow in $N^-(\p)$  \tcp*{surpluses change similarly} 
	$\delta \ot \max_i s(i)$\; 
	$\tB \ot$ Set of buyers with surplus $\delta$\tcp*{$\delta > 0$}
	$\tG \ot \{k \in G\ |\ f_{ki} > 0, i\in \tB\}$\;
Set $\gamma\ot 1$, define $p_j \ot \gamma \cdot p_j, \forall j\in \tG$, and adjust active prices and budgets accordingly as a function of $\gamma$\;
Decrease $\gamma$ continuously down from 1 until one of these events occurs:\;\label{line:priceDecStart} 
\ \ \ \ \ {\bf Event 1:} An uncapped buyer becomes capped\;
\ \ \ \ \ {\bf Event 2:} A capped good becomes uncapped\;
\ \ \ \ \ {\bf Event 3:} A new MBB edge appears\;
\ \ \ \ \ {\bf Event 4:} A subset of $\tB$ becomes \emph{tight}\tcp*{$N^-(\p)$ is feasible}\label{line:priceDecEnd}
$\f \ot$ feasible flow in $N^-(\p)$\;
$(\f, \p)\ot$ MinPrices$(\f, \p)$\;
}
$\x \ot$ FindAllocation$(\f, \p)$\; 
\Return $(\x, \p)$
\end{algorithm}

\begin{algorithm}[t]
\caption{MinPrices \label{alg:6}}
\DontPrintSemicolon
\SetKwInOut{Input}{Input}\SetKwInOut{Output}{Output}
\Input{Market $\CM$, prices $\p$, flow $\f$}
\Output{Minimum prices consistent with input configuration, feasible money flow}
$E \ot$ Set of MBB edges at prices $\p$\;
$G_c \ot$ Set of capped goods at $(\f,\p)$\;
$B_c \ot$ Set of capped buyers at $(\f,\p)$\;
Solve the following LP in price variables $\q$ and flow variables $\g$:\;
\vspace{0.1cm} \nonl \ \ \ \ 
\begin{tabular}{|ll|}
\hline
$\min \sum_j q_j$ & \\
$u_{ij}q_k =  u_{ik}q_j,$ & for each pair of edges $(i, j), (i,k) \in E$\\
$u_{ij}q_k \ge u_{ik}q_j,$ & for each pair of edges $(i,j) \in E$ and $(i,k)\not\in E$\\
$q_j \le d_j,$ & $\forall j \in G\setminus G_c$\\
$q_j \ge d_j,$ & $\forall j \in G_c$\\
$\sum_i g_{ij}= d_j,$ & $\forall j\in \G_c$\\
$\sum_i g_{ij}= q_j,$ & $\forall j \in G\setminus \G_c$\\
$\sum_j g_{ij}\ge m_i,$ & $\forall i\in B\setminus \B_c$\\
$\sum_j g_{ij}\ge c_iq_j/u_{ij},$ & $\forall i\in \B_c$ where $(i,j) \in E$\\
$m_i \ge c_iq_j/u_{ij},$ & $\forall i\in \B_c$ where $(i,j) \in E$\\
$g_{ij} = 0$, & $\forall (i,j) \not\in E$\\
$q_j \ge 0;\ g_{ij} \ge 0$ & $\forall i\in B, j\in G$\\
\hline
\end{tabular} \; 
\vspace{0.1cm}
\Return{Optimal solution $(\g^*, \q^*)$ of above LP}
\end{algorithm}

\begin{algorithm}
\caption{FindAllocation \label{alg:7}}
\DontPrintSemicolon
\SetKwInOut{Input}{Input}\SetKwInOut{Output}{Output}
\Input{Market $\CM$, prices $\p$, flow $\f$}
\Output{Allocation $\x$}
$\tG \ot \{j\in G\ |\ p_j = 0\}$\;
$\tB \ot \{i\in B\ |\ u_{ij} > 0,\ j\in \tG\}$\;
Solve the following feasibility LP in allocation variables $(x_{ij})_{i\in \tB, j\in\tG}$:\;
\vspace{0.1cm}\nonl \ \ \ \ 
\begin{tabular}{|ll|}
\hline
$\sum_{j\in \tG} u_{ij}x_{ij} = c_i$, & $\forall i\in\tB$\\
$\sum_{i\in \tB} x_{ij} \le 1,$ & $\forall j\in \tG$\\
$x_{ij} \ge 0$ & $\forall i\in\tB, j\in\tG$\\
\hline
\end{tabular}\;
\vspace{0.1cm}
$x_{ij} \ot f_{ij}/p_j, \forall i \in B \setminus \tB, j \in G \setminus \tG$\;
\Return{$\x$}
\end{algorithm}
In the body of the while-loop, we first compute a balanced flow $\f$. A \emph{balanced flow} is a maximum feasible flow in $N^-(\p)$ which minimizes the $2$-norm of surplus vector $\s = (s(1), s(2), \dots, s(|B|))$. The notion of balanced flow was introduced in~\cite{DevanurPSV08} for equilibrium computation in linear Fisher markets. It can be computed by $n$ maxflow computations. Consider two buyers $i$ and $k$ with different surplus, say $s(i) > s(k)$. If there is a good $j$ connected to $i$ and $k$ by MBB edges, then there is no flow from $j$ to $i$. Otherwise, we could decrease $f_{ji}$, increase $f_{jk}$ by the same amount, and thus decrease the $2$-norm of the surplus vector.

Let $\delta$ be the maximum surplus of any buyer, and let $\tB$ be the set of buyers with surplus $\delta$. We let $\tG$ be the set of goods $k$ that have non-zero flow to some buyer in $\tB$. 
We then decrease the prices of all goods in $\tG$ by a common factor $\gamma$. Starting with $\gamma = 1$, we decrease it continuously. This may destroy MBB edges connecting buyers in $\tB$ with goods in $G \setminus \tG$, but, by definition of $\tG$, there is no flow on such edges. For uncapped goods and capped buyers, this decreases the active price, respectively budget by a factor of $\gamma$. 
We stop if one of four events happens: (1) an uncapped buyer becomes capped, (2) a capped good becomes uncapped, (3) a new MBB edge appears, 
and (4) a subset of $\tB$ becomes tight. A subset $T$ of buyers is called \emph{tight} with respect to prices $\p$ if $\sum_{i\in T} m_i^a = \sum_{j\in \Gamma(T)} p_j^a$, where $\Gamma(T) \subseteq \tG$ is the set of goods connected to $T$ in the MBB graph. Observe that there is a feasible flow in $N^-(\p)$ iff we have 
\[
\sum_{i\in S} m_i^a \le \sum_{j\in \Gamma(S)}p_j^a, \hspace{0.5cm} \forall S\subseteq B.
\]

Next we obtain a feasible flow $\f$ in $N^-(\p)$, which is guaranteed by Event 4. 
We then use an LP (Algorithm~\ref{alg:6}) to compute the pair $(\g,\q)$ of flow and prices which minimizes $\sum_j q_j$ subject to the constraints that (1) the same buyers are capped, (2) the same goods are capped, and (3) the same edges are MBB as with respect to $(\f,\p)$. 
Since the ratio of any two prices in a connected component of the MBB graph is constant and $\f$ is a feasible solution to the LP, $\q \le \p$ component-wise. These observations imply:
\begin{lemma}\label{lem:invPLS}
The Invariants hold during the run of Algorithm~\ref{alg:PLS}. 
\end{lemma}
Finally, we find an equilibrium allocation using Algorithm~\ref{alg:7}. Here, we first obtain the set $\tG$ of zero-priced goods and the set $\tB$ of buyers who have non-zero utilities for some good in $\tG$. Clearly, the buyers of $\tB$ must be capped. For the buyers and goods in $\tB$ and $\tG$ respectively, we find an allocation by solving a feasibility LP which allocates each buyer $i$ a bundle of goods worth $c_i$ amount of utility. Note that this feasibility LP is non-empty, because we always maintain all the Invariants (Lemma~\ref{lem:invPLS}) throughout the algorithm.

We call the tuple $(E, B_c, G_c)$ a \emph{configuration}, where $E\subseteq B\times G$ is a set of MBB edges, $B_c\subseteq B$ is a set of capped buyers, and $G_c\subseteq G$ is a set of capped sellers. At the beginning of each iteration, we have a configuration based on the current prices. The following lemma ensures that our algorithm makes progress towards an equilibrium.
\begin{lemma}\label{lem:conf}
During the run of Algorithm~\ref{alg:PLS}, no configuration repeats. 
\end{lemma}
\begin{proof}
An iteration ending with Event 1 or 2 grows the set $B_c$ of capped buyers or the set $G_c$ of capped goods. Since these sets never loose members, no preceding configuration can repeat. If the sum of prices is strictly decreased before an event, i.e., $\gamma < 1$, none of the preceding configuration can repeat, since we find the minimum possible prices for the current configuration at the end of each iteration. We will show below that prices of the goods in $\tG$ are strictly decreased when an iteration ends with Event 3 or 4. 

In case of Event 3, a new MBB edge appears from a buyer $k$ in $B\setminus \tB$ to a good $j$ in $\tG$. For such an edge to become MBB, $\gamma$ must be strictly less than $1$: Since $k\not\in\tB$, $s(k) < \delta$ in the balanced flow. Suppose $\gamma=1$ then using this MBB edge from $k$ to $j$, we can increase the surplus of $k$ and decrease the surplus of a buyer in $\tB$. This decreases the 2-norm of the surplus vector, a contradiction. 

Next consider an iteration that ends due to Event 4, and suppose prices of goods in $\tG$ are not decreased. Note that the surplus of each buyer in $\tB$ is $\delta > 0$ and the surplus of each good is $0$. Hence, before we decrease prices in lines \ref{line:priceDecStart}-\ref{line:priceDecEnd} of Algorithm~\ref{alg:PLS}, we have
\[
\sum_{j\in\tG} p_j^a - \sum_{i\in\tB} m_i^a  = \delta\cdot|\tB|\enspace.
\]
When Event 4 occurs, a subset 
$T \subseteq \tB$ becomes tight, i.e., $\sum_{j\in\Gamma(T)}p_j^a - \sum_{i\in T} m_i^a = 0$, where $\Gamma(T)$ is the set of goods connected to $T$. However, $\delta\cdot|T| = \sum_{i \in T} s(i) + \sum_{j\in \Gamma(T)} s(j) = \sum_{j\in\Gamma(T)}p_j^a - \sum_{i\in T}m_i^a - \sum_{i\not\in T, j\in\Gamma(T)} f_{ji} = - \sum_{i\not\in T, j\in\Gamma(T)} f_{ji}$, which is a contradiction. 
\end{proof}
\begin{theorem}\label{thm:pls}
Algorithm~\ref{alg:PLS} computes in exponential time a thrifty and modest equilibrium in money-clearing markets. 
\end{theorem}

\begin{proof}
In each iteration, the balanced flow can be obtained in polynomial time \cite{DevanurPSV08}. Consider the maximum $\gamma$ at which an event occurs. The maximum $\gamma$ for the first three events can be easily obtained in polynomial time. For Event 4, we need to find the maximum $\gamma$ when a set of buyers become tight, which can be computed using at most a linear number of max-flow computations~\cite{DevanurPSV08,BeiGHM16}. Finally, the LP in Algorithm~\ref{alg:6} can be solved in polynomial time, hence each iteration can be implemented such that it needs only polynomial time. 

Due to Lemma~\ref{lem:conf}, we have a different configuration at the beginning of each iteration. The number of distinct configurations is finite, so Algorithm~\ref{alg:PLS} terminates with an equilibrium. The running time depends polynomially on $n$, $m$, $U$, and the number of distinct configurations, which is at most $2^{O(n\cdot m \cdot (nm))}$.
\end{proof}

Observe that our algorithm constructs an initial configuration in polynomial time. Then, for each configuration, we can interpret the sum of consistent prices as objective function, which can be found by algorithm MinPrices. Furthermore, we can define a suitable neighborhood among configurations. Algorithm~\ref{alg:PLS} can be adapted to efficiently search the neighborhood for a configuration that decreases the sum of consistent prices. Also, we can compute in polynomial time an equilibrium for a market $\CM^{s}$ as a starting configuration for our algorithm. As such, our algorithm implements the oracles for the class~\classPLS. 

\begin{corollary}
	The problem of computing a thrifty and modest equilibrium in money-clearing markets is in the class \classPLS.
\end{corollary} 
\begin{proof}
We call a configuration \emph{feasible} if the LP in Algorithm~\ref{alg:6} is feasible and its output makes the feasibility-LP of Algorithm~\ref{alg:7} non-empty. Otherwise, we call the configuration infeasible. For membership in~\classPLS, we construct polynomial-time computable neighborhood and cost functions on the set of configurations such that the following property holds: A configuration has lowest cost among all its neighbors (local optimum) if and only if it is an equilibrium. 

For each feasible configuration, let the cost be the optimum value of the corresponding LP. For each infeasible configuration, we define its cost to be prohibitively high $(n+m)U^{n+m+1}\sum_{i\in B} m_i$. Each infeasible configuration has a unique neighbor the starting configuration of Algorithm \ref{alg:PLS} (in line 1). Observe that we can take any feasible configuration as the starting configuration in Algorithm \ref{alg:PLS}. Hence, we define the unique neighbor of each feasible configuration $\mathcal C$ as the next configuration in Algorithm \ref{alg:PLS} when it is started with $\mathcal C$. Clearly, both cost and neighborhood functions are polynomial-time computable, and a configuration is a local optimum if and only if it is a thrifty and modest equilibrium. This proves the claim. 
\end{proof}

\begin{remark} \rm 
It is not clear how to use Algorithm~\ref{alg:FPTAS} to show membership in class \classPLS. The difficulty lies in defining a suitable configuration space and a potential function.
\end{remark}

\renewcommand{\l}{{\bm{\lambda}}}

\subsection{Constant Number of Buyers or Goods}
\label{sec:Constant}

In this section, we show that Algorithm~\ref{alg:PLS} runs in polynomial time when either the number of buyers or the number of sellers is constant. Consider the number of MBB graphs for a fixed set of capped buyers and capped sellers. Using a cell decomposition technique, we show it is polynomial when $|B|$ or $|G|$ is constant. We create regions in a constant dimensional space by introducing polynomially many hyperplanes. The number of non-empty regions formed by $N$ hyperplanes in $\mathbb R^d$ is $O(N^d)$. Thus, we get a polynomial bound on the number of regions. 

Next we show that each MBB graph maps to a particular region thus created. Since the number of regions is polynomial, we get a polynomial bound on the number of different MBB graphs. This implies that for any given set of capped buyers and capped sellers, Algorithm~\ref{alg:PLS} examines only polynomially many configurations. Since the set of capped buyers only grows and the set of capped sellers only shrinks, this implies a polynomial running time for Algorithm~\ref{alg:PLS}. 

\begin{theorem}
Algorithm~\ref{alg:PLS} computes in polynomial time a thrifty and modest equilibrium in money-clearing markets with constant number of buyers or sellers. 
\end{theorem}
\begin{proof}
For constant number of goods, consider the following set of hyperplanes in $(p_1, \dots, p_{|G|})$-space, where $p_j$ denotes the price of good $j$. 
\[u_{ij}p_{j'} - u_{ij'}p_j = 0, \forall i\in B, \forall j, j' \in G.\]
These hyperplanes partition the space into cells, and each cell has one of the signs $<, =, >$ for each hyperplane. Further, each MBB graph $(B\cup G, E)$ satisfies the following constraints in $\p$ variables: 
\begin{eqnarray}
\forall (i,j),(i, j') \in E & : & u_{ij}p_{j'} - u_{ij'}p_{j} = 0\notag\\
\forall (i,j) \in E\ \&\ \forall (i,j')\not\in E & : & u_{ij}p_{j'} - u_{ij'}p_j \ge 0.\notag
\end{eqnarray}
Now, for constant number of buyers consider the following set of hyperplanes in $(\lambda_1, \dots, \lambda_{|B|})$-space, where $1/\lambda_i$ denotes the MBB of buyer $i$. 
\[\lambda_i u_{ij} - \lambda_{i'}u_{i'j} = 0, \forall i, i' \in B, \forall j\in G.\]
These hyperplanes partition the space into cells, and each cell has one of the signs $<, =, >$ for each hyperplane. Further, each MBB graph $(B\cup G, E)$ satisfies the following constraints in $\l$ variables: 
\begin{eqnarray}
\forall (i,j),(i', j) \in E & : & \lambda_i u_{ij} - \lambda_{i'}u_{i'j} = 0\notag\\
\forall (i,j) \in E\ \&\ \forall (i',j)\not\in E & : & \lambda_i u_{ij} - \lambda_{i'}u_{i'j} \ge 0\notag
\end{eqnarray}

In both cases, each MBB graph maps to a particular cell in the cell decomposition. Since the number of cells are polynomially bounded for constant $|G|$ or $|B|$, this implies a polynomial bound on the number of different MBB graphs. Thus, since the set of capped buyers only grows and the set of capped sellers only shrinks, we get a polynomial running time for Algorithm~\ref{alg:PLS}.  
\end{proof}

\renewcommand{\l}{{\bm{\lambda}}}
\renewcommand{\L}{{\Lambda}}
\newcommand{\ttb}{{\tiny\textbullet \ \ }}

\renewcommand{\P}{{\p^{min}}}
\newcommand{\LCP}{{\sf LCP }}

\subsection{Membership in PPAD}
\label{sec:PPAD}
In this section, we show that computing a thrifty and modest equilibrium in money-clearing markets $\CM$ is in the class \classPPAD. We first derive a formulation as a linear complementarity problem (LCP). It captures all thrifty and modest equilibria of $\CM$, but also has non-equilibrium solutions. To discard the non-equilibrium solutions, we incorporate a positive lower bound on variables representing prices of goods and MBB of buyers. This turns out to be a non-trivial adjustment, because a subset of prices may be zero at all equilibria, so we must be careful not to discard equilibrium solutions. Our approach is based on our previous work~\cite{BeiGHM16}, in which we gave a polynomial-time algorithm for markets $\CM^b$ with utility limits (and without earning limits) to find an equilibrium, whose prices are coordinate-wise smallest among all equilibria. 

Our approach can be summarized as follows. Consider a money-clearing market $\CM$. Now suppose we remove all earning caps to obtain a market $\CM^b$. To this market we apply the algorithm of~\cite{BeiGHM16} and obtain a min-price equilibrium $(\x^{min},\P)$. We show that using $\p^{min}$, the market $\CM$ can be partitioned into two separate markets $\CM_1$ and $\CM_2$. Market $\CM_1$ consists of all goods with price 0 in $\P$ and all buyers having non-zero utility for these goods. $\CM_2$ consists of the remaining buyers and goods. Since all buyers in $\CM_2$ have no utility for goods in $\CM_1$, we have that $\CM_2$ is money clearing if and only if $\CM$ is money clearing.

Based on these two markets, we show that there is\footnote{In fact, it can be shown that every equilibrium of $\CM$ has this property, but this is not necessary for membership in \classPPAD.} an equilibrium of $\CM$ that is an equilibrium of $\CM_1$ and an equilibrium of $\CM_2$. We already know an equilibrium for $\CM_1$ with price of 0 for every good. For $\CM_2$ we show that at every equilibrium, the price of a good $j$ is at least $p^{min}_j$. Using this lower bound on the equilibrium prices in $\CM_2$, we construct a modified LCP formulation $\CM$-\LCP which exactly captures all equilibria of $\CM_2$. We suitably add an auxiliary scalar variable to $\CM$-\LCP and apply Lemke's algorithm. If $\CM_2$ is money clearing, we show that Lemke's algorithm is guaranteed to converge to an equilibrium of $\CM_2$. Composing this with the equilibrium of $\CM_1$ gives an equilibrium of $\CM$. Further, using a result of Todd~\cite{Todd76}, this proves that computing an equilibrium in money-clearing markets $\CM$ is in \classPPAD.

\subsubsection{LCP Formulation}
We start our analysis by deriving an LCP formulation to capture equilibria of $\CM$. The LCP has the following variables 

\ttb $\p=(p_j)_{j\in G}$, where $p_j$ is the price of good $j$,

\ttb $\f=(f_{ij})_{i\in B, j\in G}$, where $f_{ij}$ is the money spent on good $j$ by buyer $i$,

\ttb $\l=(\li)_{i\in B}$, where $1/\li$ is the MBB of buyer $i$ at prices $\p$,

\ttb $\bDelta=(\delta_i)_{i\in B}$, where $(m_i - \delta_i)$ is the active budget of buyer $i$,

\ttb $\bBeta=(\beta_j)_{j\in G}$, where $(p_j - \beta_j)$ is the active price of good $j$.\\

Let $\perp$ denote a complementarity constraint between the inequality and the variable (e.g., $u_{ij}\li - p_j \le 0 \perp f_{ij} \ge 0$ is a shorthand for $u_{ij}\li - p_j \le 0;\ f_{ij} \ge 0;\ f_{ij}(u_{ij}\li - p_j) = 0$). 

\begin{eqnarray}
\forall (i,j)\in (B, G): &  u_{ij}\li - p_j \le 0  &  \perp  \ \ f_{ij} \ge 0 \label{eqn:opt}\\
\forall i\in B: & \delta_i \ge m_i - c_i\li & \perp \ \ \delta_i\ge 0 \label{eqn:alpha}\\
\forall i\in B: & -\sum_{j}f_{ij} \le - (m_i - \delta_i) & \perp \ \  \li\ge 0 \label{eqn:mcb}\\
\forall j\in G: & \beta_j \ge  p_j - d_j & \perp \ \ \beta_j \ge 0 \label{eqn:compg} \\
\forall j\in G: & \sum_{i}f_{ij} \le p_j-\beta_j  & \perp \ \  p_j\ge 0 \label{eqn:mcg}
\end{eqnarray}

\begin{lemma}\label{lem:capeq}
The LCP defined by \eqref{eqn:opt}-\eqref{eqn:mcg} captures all equilibria of $\CM$.  
\end{lemma}

\begin{proof}
Let $(\f,\p)$ be an equilibrium of $\CM$. Let $1/\li$ capture the MBB of buyer $i$ at prices $\p$. Clearly, $\li = \min_{j: u_{ij} > 0} p_j/u_{ij}$. From the optimal bundle constraint, it is clear that $f_{ij} > 0$ only if $u_{ij}\li - p_j = 0$. This implies that $(\p, \f, \l)$ satisfies \eqref{eqn:opt}. The active budget $m_i^a$ of buyer $i$ is $\min\{m_i, c_i\li\}$ and we have $\sum_j f_{ij} = m_i^a$. Setting $\delta_i = \max\{0, m_i - c_i\li\}$ satisfies \eqref{eqn:alpha}. This further implies that $m_i^a = m_i - \delta_i$, and we get \eqref{eqn:mcb}. Similarly, the active price $p_j^a$ of good $j$ is $\min\{p_j, d_j\}$ and we have $\sum_i f_{ij} = p_j^a$. Setting $\beta_j = \max\{0, p_j - d_j\}$ satisfies \eqref{eqn:compg}. This further implies that $p_j^a  = p_j - \beta_j$ and we also get \eqref{eqn:mcg}. This proves the claim. 
\end{proof}

Lemma \ref{lem:capeq} shows that all equilibria of $\CM$ are captured by the LCP \eqref{eqn:opt}-\eqref{eqn:mcg}. However, there are solutions to this LCP which are not market equilibria, e.g., $\li = p_j = \beta_j = f_{ij} = 0$, $\forall i,\forall j$, and $\delta_i = m_i, \forall i$ is an LCP solution but not an equilibrium. To discard these non-equilibrium solutions in the LCP, we strive to include a positive lower bound to all $p_j$ and $\li$. 

\begin{remark} \rm
Previous constructions~\cite{GargMSV15,GargMV18} of LCPs for market equilibria use only a positive lower bound of $1$ to the prices. In our case this is not sufficient, since there are solutions where $\li =  f_{ij} = 0, \forall(i,j)$, and $\delta_i = m_i, \forall i$. As a consequence, we need to establish positive lower bounds to all $\li$. Another difficulty arises from the fact that a positive constant cannot be used as a lower bound, because there might be prices that are zero in all equilibria. A positive lower bound for these prices would discard all equilibria as solutions of the LCP.
\end{remark}

To handle these difficulties we use our polynomial-time algorithm~\cite{BeiGHM16} for markets $\CM^b$ with utility caps. Consider market $\CM$ and disregard all earning caps. The resulting market is a market $\CM^b$, for which our algorithm from~\cite{BeiGHM16} can compute a price vector $\p^{min}=(p^{min}_j)_{j\in G}$ of a min-price equilibrium. A min-price equilibrium has coordinate-wise smallest prices, i.e., for every good $j$ the price $p^{min}_j$ is the smallest price of good $j$ in all equilibria. As a consequence, the set $S=\{j\in G\ |\ p^{min}_j = 0\}$ includes all goods that have price zero in every equilibrium of market $\CM^b$. Let $\Gamma(S)$ be the set of buyers who derive non-zero utility from goods in $S$, i.e., $\Gamma(S) = \{i\in B\ |\ u_{ij} > 0, j\in S\}$. 

We partition the market into two disjoint markets. Market $\CM_1$ includes exactly the goods of $S$ and the buyers in $\Gamma(S)$. Market $\CM_2$ has the remaining goods and buyers. By definition $u_{ij} = 0$ for every $j \in S$ and $i \not\in \Gamma(S)$ and hence no buyer in $\CM_2$ will ever spend on goods in $\CM_1$. The min-price equilibrium $(\vecx^{min}, \vecp^{min})$ yields an equilibrium for $\CM_1$, since all utility caps for all $i \in \Gamma(S)$ are reached and all earning caps for all $j \in S$ are satisfied. We will next establish that there is an equilibrium for $\CM_2$ in which every good $j \not\in S$ has a price $p_j \ge p^{min}_j$ for all $j \not\in S$. Then no buyer $\CM_1$ will ever spend on goods from $\CM_2$, justifying the separation of the markets.

\begin{lemma}\label{lem:gtp} In every equilibrium for market $\CM_2$, $p_j \ge p^{min}_j, \forall j\not\in S$.   
\end{lemma}

\begin{proof}
Suppose there is an equilibrium of $\CM_2$ where the price of some good $j$ is $p_j < p^{min}_j$. Note that at $\p^{min}$, the sum of prices of goods in $\CM_2$ is exactly equal to the sum of active budgets of buyers in $\CM_2$. Since $(\x^{min},\p^{min})$ is a min-price equilibrium, there is an uncapped buyer in every MBB component of $\CM_2$. Let $\gamma = \min_{k\in G} p_k/p^{min}_k$ and let $j$ be a good for which $p_j/p^{min}_j = \gamma$. Now consider an MBB component $C$ containing good $j$ at $\p$. Since $\gamma <1$ and there is an uncapped buyer in every MBB component of $\CM_2$ at $\p^{min}$, we can conclude that the total prices of goods in $C$ will be less than the total active budgets of buyers in $C$, which is a contradiction. 
\end{proof}

Next we derive an LCP for market $\CM_2$ using the lower bound on the price of each good as given in Lemma~\ref{lem:gtp}. At this point, we need to solve the equilibrium problem for $\CM_2$ only, so let $B_2 = B \setminus \Gamma(S)$ and $G_2 = G \setminus S$ denote the sets of buyers and goods in $\CM_2$, respectively. For the lower bound on $\li$'s we define
\begin{equation}\label{eqn:defL}
\Lambda \defeq \nicefrac{1}{2}\min_{i \in B_2, j \in G_2:\ u_{ij} > 0} \{p^{min}_j/u_{ij}\}.
\end{equation}

Consider the following modified LCP in variables $(\l, \p', \f, \bDelta, \bBeta)$, where price of good $j$ is $p'_j + p^{min}_j$:
\begin{eqnarray}
\forall (i,j)\in (B_2, G_2): &  u_{ij}(\li + \L) - (p'_j + p^{min}_j) \le 0  &  \perp  \ \ f_{ij} \ge 0 \label{eqn:opt2}\\
\forall j\in G_2: & \sum_{i}f_{ij} \le p'_j + p^{min}_j -\beta_j  & \perp \ \  p'_j\ge 0 \label{eqn:mcg2}\\
\forall i\in B_2: & -\sum_{j}f_{ij} \le - (m_i - \delta_i) & \perp \ \  \li\ge 0 \label{eqn:mcb2}\\
\forall i\in B_2: & \delta_i \ge m_i - c_i(\li + \L) & \perp \ \ \delta_i\ge 0 \label{eqn:alpha2}\\
\forall j\in G_2: & \beta_j \ge  p'_j + p^{min}_j - d_j & \perp \ \ \beta_j \ge 0 \label{eqn:compg2} 
\end{eqnarray}

The constraints \eqref{eqn:opt2}-\eqref{eqn:compg2} represent the $\CM$-LCP. Next we show that this LCP exactly captures all market equilibria of $\CM_2$. 
\begin{lemma}
A solution of $\CM$-LCP is a thrifty and modest equilibrium of $\CM_2$ and vice-versa. 
\end{lemma}

\begin{proof}
Lemmas \ref{lem:capeq} and \ref{lem:gtp} show that every equilibrium of $\CM_2$ is a solution of $\CM$-LCP. For the other direction, consider a solution $(\l, \p', \f, \bDelta, \bBeta)$ of $\CM$-LCP. The active price $p_j^a$ of good $j$ is $p'_j + p^{min}_j - \beta_j$ and the active budget $m_i^a$ of buyer $i$ is $m_i - \delta_i$. Clearly, $m_i^a > 0, \forall i \in B_2$ and $p_j^a > 0, \forall j\in G_2$. 

Next we claim that $\li > 0, \forall i$. Suppose $\li = 0$ for a buyer $i$, then $f_{ij} = 0, \forall j$ due to \eqref{eqn:opt2} and \eqref{eqn:defL} which violates the left inequality of \eqref{eqn:mcb2}. Hence $\li > 0, \forall i\in B_2$. 

The constraints \eqref{eqn:opt2} ensure that $f_{ij} > 0$ then $u_{ij}/(p'_j + p^{min}_j) = \max_{k\in G_2} u_{ik}/(p'_k + p^{min}_k)$, which implies that each buyer buys an optimal bundle. 

Further, the constraints \eqref{eqn:mcb2} together with the fact that $\li > 0, \forall i$ ensure that each buyer spends its entire active budget. Now we only need to show that each good receives money equal to its active price, i.e., $\sum_i f_{ij} = p_j^a, \forall j\in G_2$. 

Note that the prices $\p^{min}$ impose an equilibrium for $\CM_2$ without the earning caps. Let $S' = \{j\in G_2 \ |\ p_j' = 0\}$. Clearly, $\sum_i f_{ij} = p_j^a, \forall j\in G_2 \setminus S'$ due to \eqref{eqn:mcg2}. Let $\Gamma(S')$ be the set of buyers having at least one MBB good in $S'$ at prices $\p^{min}$.  Since $\p^{min}$ is the min-price equilibrium without earning caps, the total active budget of buyers in $\Gamma(S')$ is at least the total prices of goods in $S'$, i.e., $\sum_{j\in S'} p^{min}_j$. 

Now suppose we set prices $p_j = p'_j + p^{min}_j \ge p^{min}_j, \forall j \in G_2$. Since $p_j > p^{min}_j, \forall j\in G_2\setminus S'$ and $p_j = p^{min}_j, \forall j\in S'$, buyers in $\Gamma(S')$ have all their MBB goods in $S'$ at prices $\p$ and the active budget of the buyers in $\Gamma(S')$ is at least $\sum_j p^{min}_j$. Thus, we have 
\begin{equation}\label{eqn:mcs}
\sum_{j\in S'} p^{min}_j \le \sum_{i\in\Gamma(S')} m_i^a.
\end{equation}
Further, summing the left hand side inequality of constraints~\eqref{eqn:mcg2} and~\eqref{eqn:mcb2} for buyers in $\Gamma(S')$ and using the fact that $\li > 0, \forall i$, we get 
\[\sum_{i\in \Gamma(S')} m_i^a = \sum_{i\in \Gamma(S'),j\in S'} f_{ij} \le \sum_{j\in S'} p_j^a = \sum_{j\in S'} p^{min}_j - \sum_{j\in S'} \beta_j \]
The above with \eqref{eqn:mcs} imply that $\beta_j = 0, \forall j\in S'$ and all inequalities are equalities. Hence we have $\sum_i f_{ij} = p_j^a, \forall j\in G_2$. 
\end{proof}

\subsubsection{Lemke's Algorithm} 
In this section, we apply Lemke's algorithm (see Appendix~\ref{app:lcp} for details) on the $\CM$-LCP. For this, we first add an auxiliary, non-negative, scalar variable $z$ in \eqref{eqn:mcb} and consider 
\begin{equation}
-\sum_{j}f_{ij} - z \le - (m_i - \delta_i) \ \ \perp \ \ \li\ge 0  \ \ \ \text{ and } \ \ \  z \ge 0 \label{eqn:mcbz}
\end{equation}
We denote by $\CM$-LCP2 the constraints (\ref{eqn:opt2}-\ref{eqn:mcg2}), \eqref{eqn:mcbz}, (\ref{eqn:alpha2}-\ref{eqn:compg2}). The primary ray of $\CM$-LCP2 is $z \ge m_i-\delta_i, \forall i$. The other variables are set to $\l = \0$, $\p'=\0$, $\f = \0$, $\delta_i = \max\{0, m_i-\L c_i\}, \forall i$ and $\beta_j = \max\{0,p^{min}_j - d_j\}, \forall j$. In the proof of the following theorem we show that under the money-clearing condition, there are no secondary rays in $\CM$-LCP2. Hence, applied to this LCP Lemke's algorithm will converge to an equilibrium. 

\begin{theorem}\label{thm:conv}
Lemke's algorithm applied to $\CM$-LCP2 converges to an equilibrium in money-clearing markets.
\end{theorem}

\begin{proof}
We prove the result by contradiction. Suppose Lemke's algorithm converges on a secondary ray $R$, which starts at a vertex $(\l_*, \p'_*, \f_*, \bDelta_*, \bBeta_*, z_*)$ where $z_* > 0$ and the direction vector is $(\l_o, \p'_o, \f_o, \bDelta_o, \bBeta_o, z_o)$, i.e., 
\[R = \{(\l_*, \p'_*, \f_*, \bDelta_*, \bBeta_*, z_*) + \alpha(\l_o, \p'_o, \f_o, \bDelta_o, \bBeta_o, z_o), \forall \alpha \ge 0\}.\] 
Observe that $(\l_o, \p'_o, \f_o, \bDelta_o, \bBeta_o, z_o) \ge 0$. We consider three cases and show a contradiction in each of them. 
\begin{description}
\item[Case 1:] In this case $\p'_o > \0$, i.e., all prices are increasing on $R$. Since $p_j^a = p'_j + p^{min}_j - \beta_j = \min\{p'_j + p^{min}_j, d_j\}$, $p_j^a = d_j, \forall j$ on $R$. Further, $\sum_{i} f_{ij} = p_j^a, \forall j$ on $R$. Note that if $\li = 0$, then $f_{ij} = 0, \forall j$ due to \eqref{eqn:opt}. This further implies that $f_{ij} > 0$ only if $\li > 0$. Hence, using \eqref{eqn:mcbz} we get $\sum_{i\in S} (m_i - \delta_i) - z|S| = \sum_j d_j$, where $S = \{i\in B_2 \mid \li > 0\}$. By money clearing~\eqref{eqn:nsc}, $z = 0$ and hence $z_* = 0$; a contradiction.
\item[Case 2:] In this case $\p'_0 = \0$, i.e., all prices are constant on $R$. Since $\p'_o = \0$, we have $\l_o = \0$ and $\f_o = \0$ due to~\eqref{eqn:opt} and \eqref{eqn:mcg}, respectively. This further implies that $\bDelta_o = \0$ and $\bBeta_o = \0$. Together, they imply that $z_o > 0$ because the direction vector cannot be $\0$. Since $z_o > 0$, we get $\l_* = \0$ due to~\eqref{eqn:mcbz}, which subsequently implies $\f_* = \0$, $\p'_* = \0$, $(\bDelta_*)_i = \max\{0, m_i-c_i\L\}$, and $(\bBeta_*)_j = \max\{0, p^{min}_j - d_j\}$. This means that the secondary ray is a primary ray; a contradiction. 
\item[Case 3:] In this case $\p'_o \ngtr \0$ and $\p'_o \neq \0$, i.e., some prices are increasing and some are constant on $R$. Let $S' = \{j \in G_2 \mid (p'_o)_j > 0\}$. This implies that $p_j^a = d_j$ and $\sum_{i} f_{ij} = d_j$ $\forall j\in S'$ due to~\eqref{eqn:mcg2}. Prices of goods in $S'$ are increasing to infinity on $R,$ and these goods are sold upto their maximum possible revenue. Hence, the buyers who buy these goods have zero utility for the goods outside $S'$, because each buyer buys an optimal bundle, and the prices of goods outside $S'$ are constant on $R$. Let $\Gamma(S')$ be the set of buyers buying goods in $S'$ on $R$. This implies that $\li > 0, \forall i\in \Gamma(S')$ due to~\eqref{eqn:opt}. Using \eqref{eqn:mcbz} we get that $\sum_{i\in\Gamma(S')} (m_i - \delta_i) - z|\Gamma(S')| = \sum_{j\in S'} d_j$. This implies that $z=0$ using the money-clearing condition for the buyers in $\Gamma(S')$; a contradiction. \qedhere
\end{description}
\end{proof}

\begin{corollary}
The problem of computing a thrifty and modest equilibrium in money-clearing markets is in the class \classPPAD. 
\end{corollary}

\begin{proof}
By Theorem \ref{thm:conv}, Lemke's algorithm must converge to an equilibrium for money-clearing markets $\CM$. Note that Lemke's algorithm traces a path on the $1$-skeleton of a polyhedron. Let $v$ be a vertex on the path found by Lemke's algorithm. To prove membership in \classPPAD, we need to show that the unique predecessor and successor of $v$ on this path can be found efficiently. Clearly, these two vertices, say $u$ and $w$, can be found simply by pivoting. To determine which vertex leads to the start of the path, i.e., the primary ray, and which leads to the end, we use a result by Todd~\cite{Todd76} on the orientability of the path followed by a complementary pivot algorithm. It shows that the signs of the sub-determinants of tight constraints satisfied by the vertices $u$, $v$ and $w$ provide the orientation of the path. This concludes the proof of membership in \classPPAD.
\end{proof}

\begin{remark} \rm
We note that a money-clearing market $\CM$ can be reduced to a more general Leontief-free market~\cite{GargMV14}. However, the agents in the reduced market remain satiated because buyers and sellers in $\CM$ are thrifty. The results for Leontief-free markets in~\cite{GargMV14} (such as membership in \classPPAD) require non-satiation of agents. Hence, these results are not directly applicable to markets $\CM$ via such a reduction.
\end{remark}
}

\section{Approximating the Nash Social Welfare}
\label{sec:NSW}

\subsection{Constant-Factor Approximation for Budget-Additive Valuations}
\label{sec:NSWalgo}
In this section, we present a $(2e^{1/(2e)}+\varepsilon)$-approximation algorithm for maximizing Nash social welfare with budget-additive valuations, for every constant $\varepsilon > 0$. 

Consider maximization of Nash social welfare for allocation of a set $G$ of indivisible items to a set $B$ of agents with budget-additive valuations. As a first step, we execute a simple adjustment to the valuation functions. Note that if for some agent $i\in B$ and some item $j \in G$ we have $v_{ij} \ge c_i$, we can equivalently assume that $v_{ij} = c_i$ since the valuation of $i$ can be at most $c_i$. More formally, let $v'_{ij} = \min(v_{ij}, c_i)$ and $v'_i(\x_i^S) = \min \left( c_i, \sum_{j \in G} v'_{ij}x_{ij}^S \right)$. The following lemma is straightforward and its proof is omitted. 

\begin{lemma}
For every integral allocation $\x$ we have $v'_i(\x) = v_i(\x)$.
\end{lemma}

Henceforth, we will assume that $v_{ij} \le c_i$, for all $i \in B$, $j \in G$. 

To solve the integral maximization problem, consider the following convex program that describes a natural fractional relaxation.
\begin{equation}\label{pro:EG}
	\begin{array}[4]{rrcll}
	\text{Max.} & \multicolumn{4}{l}{\D \left(\prod_{i \in B} \D \left(\sum_{j \in G} v_{ij} x_{ij}\right) \right)^{1/n} }\vspace{0.2cm}\\
	\text{s.t.}	& \D \sum_{j \in G} v_{ij} x_{ij}& \leq & c_i & i \in B \vspace{0.2cm}\\
	          	& \D \sum_{i \in B} x_{ij} & \leq & 1 & j \in G\\
						  & x_{ij} & \ge & 0 & i \in B,\ j \in G\\
	\end{array}
\end{equation}
The optimal solution to this program is the allocation vector of a thrifty and modest equilibrium in a Fisher market, in which agent $i$ has a linear utility with $u_{ij} = v_{ij}$, a utility limit $c_i$ and a budget $m_i = 1$ (for details, see, e.g.~\cite{BeiGHM16,ColeDGJMVY17}). 

Unfortunately, the optimal fractional allocation of this program can be significantly better in terms of the Nash social welfare than any integral solution with all $x_{ij} \in \{0,1\}$ (see~\cite{ColeG18} for an example, in which the ratio is exponential in $|B|$). Hence, similar to~\cite{BeiGHM17,ColeDGJMVY17,ColeG18}, we introduce additional earning limits $d_j = 1$ for all $j \in G$ into the Fisher market. We will see that this lowers the achievable objective function value in equilibrium and allows us to round the fractional equilibrium allocation to an integral one that approximates the optimal Nash social welfare.

Consider the resulting Fisher market $\CM$ with utility and earning limits. Our first observation is that non-trivial instances of the Nash social welfare problem give rise to a market $\CM$ that is money clearing.

\begin{lemma}
  \label{lem:noMoneyNoNSW}
  Consider the Fisher market $\CM$ resulting from an instance of the Nash social welfare problem. If the market $\CM$ is not money clearing, then the maximum Nash social welfare for indivisible items is 0. 
\end{lemma}

\begin{proof}
Obviously, if market $\CM$ is not money clearing, then there exists a subset $B'$ of buyers such that the sum of earning caps of goods in $\Gamma(B') = \{j\ |\ v_{ij} > 0, i \in B'\}$ is less than the sum of budgets of buyers in $B'$. This implies that $|\Gamma(B')| < |B'|$. Hence, there is no allocation where each agent in $B'$ gets at least one item of positive valuation. Thus, the Nash social welfare must always be 0.
\end{proof}

When the market is not money clearing, every integral allocation has the optimal Nash social welfare. It is easy to check condition~\eqref{eqn:nsc} by a max-flow computation. 

Hence, for the remainder of this section, we assume that the instance of the Nash social welfare problem is non-trivial, i.e., the resulting Fisher market is money clearing. We have seen in Section~\ref{sec:exist} above that a money-clearing market $\CM$ always has a thrifty and modest equilibrium. Suppose we are given such an equilibrium $(\x, \p)$. 

The Nash social welfare objective allows scaling the valuation function of every agent $i$ by any factor $\gamma_i > 0$. This adjustment does neither change the equilibrium, the integral optimum solution of the Nash social welfare problem, nor the approximation factor. Given the equilibrium $(\x, \p)$, we want to \emph{normalize} the valuation function for agent $i$ based on the MBB ratio $\alpha_i$ of buyer $i$ in the market equilibrium. 

In equilibrium, there can be a set of goods $G_0 = \{ j \mid p_j = 0\}$. All buyers $B_0 = \{ i \mid u_{ij} > 0 \text{ for some } j \in G_0\}$ interested in any good $j \in G_0$ have infinite MBB ratio. Due to our equilibrium conditions, every $i \in B_0$ must be capped and receive allocation only from $G_0$, i.e., $u_i(\vecx) = c_i$ and $x_{ij} > 0$ only if $j \in G_0$ and $u_{ij} > 0$. Moreover, since no buyer $i \in B \setminus B_0$ has positive utility for any of the goods $G_0$, these goods are allocated only to $B_0$. Therefore, we can treat items $G_0$ and agents $B_0$ separately in the analysis.

For all $i \in B \setminus B_0$, we normalize $v'_{ij} = v_{ij}/\alpha_i$ and $c'_i = c_i/\alpha_i$. This does not change the demand bundle for buyer $i$, and thus $(\x,\p)$ remains an equilibrium. In the resulting instance, every such buyer has MBB of 1 in $(\x,\p)$. Consequently, $v'_{ij} \le p_j$ for all $i\in B \setminus B_0, j\in G$, where equality holds if and only if $j$ is an MBB good of buyer $i$. For simplicity we assume that $v$ and $c$ fulfill these conditions directly, i.e., $v_{ij} = v'_{ij}$ and $c_i = c_i'$. Together with the fact that $v_{ij} \le c_i, \forall (i,j)$ this implies 
\begin{equation}
  \label{eq:ubMatch}
  v_{ij} \le \min(p_j, c_i),\hspace{1cm} \text{ for all } i\in B \setminus B_0, j\in G\enspace.
\end{equation}
The following lemma is a helpful insight on the structure of equilibria.

\begin{lemma}\label{lem:uncap}
Consider a money-clearing Fisher market $\CM$ with $m_i = 1$, $d_j = 1$, and $v_{ij} \le c_i$, for all $i\in B, j\in G$. Suppose we normalize the utilities based on a thrifty and modest equilibrium $(\vecx, \vecp)$. Then the following properties hold.

\begin{compactenum}[\mbox{}\hspace{\parindent}(a)]
\item A buyer $i \in B \setminus B_0$ spends $m_i^a = \min(1,c_i)$ units of money. Its valuation $v_i(\x)$ is equal to $m_i^a$. If $i$ is capped, $c_i \le 1$. 
\item If $i$ is capped, it is allocated at least one unit of goods.
\item If $i$ is capped and $j$ is an MBB good for buyer $i$. Then $p_j \le 1$.
\item If $p_j < 1$, $j$ is completely sold.
\item The money spent on good $j$ is $p_j^a = \min(p_j,1)$. 
\end{compactenum}
\end{lemma}

\begin{proof} An uncapped buyer $i \in B \setminus B_0$ spends his entire budget as otherwise $\x_i$ would not be a demand bundle. Since MBB = 1, the valuation $v_i(\x)$ is equal to the money spent by $i$. If $i$ is capped, its valuation is equal to $c_i$ and hence the money spent is equal to $c_i$. 

Since $v_{ij} \le c_i$ always and $c_i = v_i(\x)$ for a capped agent, we have 
$c_i = v_i(\x) \le c_i \sum_j x_{ij}$ and hence $\sum_j x_{ij} \ge 1$. 

If $j$ is an MBB good for a capped buyer $i$, then $p_j = v_{ij} \le \min(p_j,c_i)$ according to (\ref{eq:ubMatch})
and hence $p_j \le c_i \le 1$, where the last inequality was established in  (a). 

If $0 < p_j < 1$, the supply $e_j = \min(1,d_j/p_j) = \min(1,1/p_j) = 1$. Thus $j$ is completely sold. 

Finally, the money spent on $j$ is $p_j e_j = p_j \min(1,1/p_j) = \min(p_j,1) = p_j^a$. 
\end{proof}

Our subsequent analysis proceeds as follows. First, in the following Section~\ref{sec:ubNSW}, we describe an upper bound on the optimal Nash social welfare of any integral solution. The upper bound is based on the properties of any thrifty and modest equilibrium described above. In Section~\ref{sec:algoNSW} we then show how to round an equilibrium and obtain an integral solution that is a $2e^{1/(2e)}$-approximation for the optimal Nash social welfare. Finally, since our FPTAS from Section~\ref{sec:FPTAS} computes an exact equilibrium in a perturbed market, we discuss in Section~\ref{sec:perturbNSW} the impact of perturbing the market on the approximation guarantee.

\subsubsection{Upper Bound}
\label{sec:ubNSW}
In this section, we describe an upper bound on the optimal Nash social welfare when valuations are normalized based on an equilibrium $(\x, \p)$. The bound relates to prices and utility caps of the capped buyers in $(\x,\p)$. We denote by $B_c$ and $B_u$ the set of capped and uncapped buyers in $(\x,\p)$, respectively. Recall that since $(\x,\p)$ is a thrifty and modest equilibrium, buyers may not spend their entire budget and sellers may not sell their entire supply. We denote by $m_i^a = \min(1, c_i)$ the active budget of buyer $i$ and by $p_j^a = \min(p_j, d_j)$ the active price of good $j$. The following result is a generalization of a similar bound shown in~\cite{ColeG18}. The main difference is to carefully account for the contribution of capped buyers.

\begin{theorem}
For valuations $v$ and caps $c$ normalized according to equilibrium prices $\p$, we have
\[
    \left(\prod_{i\in B} v_i(\x^*)\right)^{1/n} \le \left(\prod_{i\in B_c} c_i \prod_{j:p_j > 1} p_j\right)^{1/n}\enspace,
\]
where $x^*$ is an integral allocation that maximizes the Nash social welfare.
\end{theorem}

\begin{proof}
Consider an \emph{integral} allocation $\x^*$ that maximizes the Nash social welfare. For the agents $i \in B_0 \subseteq B_c$, a simple upper bound is $\prod_{i \in B_0} v_i(\x^*) \le \prod_{i \in B_0} c_i$. To obtain an upper bound on $\prod_{i \in B \setminus B_0} v_i(\x^*)$, we have to work harder. We denote by $G_c$ the set of goods allocated to agents in $B_c \setminus B_0$ in $\x^*$ and by $G_u^1 = \{j \in G\ |\ p_j > 1; x^*_{ij} = 1 \text{ for some } i\in B_u\}$ the set of items with price more than 1 that are assigned in $\x^*$ to buyers in $B_u$. Let $B_u^1 = \{i\in B_u \ |\ x^*_{ij} = 1 \text{ for some }  j \in G_u^1\}$ be the set of buyers from $B_u$ that receive an item of $G_u^1$ in $\x^*$. Note that $|B_u^1| \le |G_u^1|$. 

\newcommand{\abs}[1]{| #1 |}
We will construct a fractional allocation $\tilde{\x}$ with $\prod_{i\in B \setminus B_0} v_i(\x^*) \le \prod_{i\in B \setminus B_0} v_i(\tilde{\x})$ and then bound the latter product. A key step in the construction of $\tilde{\x}$ is to bound $\sum_{i \in B_u \setminus B_u^1}  v_i(\x^*) $. We have
\begin{align}
\nonumber
\sum_{i \in B_u \setminus B_u^1} v_i(\x^*) &\le \sum_{i \in B_u \setminus B_u^1} \sum_{j \in G \setminus (G_c \cup G_u^1)} p_j \x^*_{ij} \le \sum_{j \in G \setminus (G_c \cup G_u^1)} p_j \\
\nonumber
&= \sum_{i\in B_c \setminus B_0} c_i + |B_u| - |G_u^1|  - \sum_{j\in G_c} p_j^a \\
\label{eq:nubBu}
&\le \sum_{i\in B_c \setminus B_0} c_i + |B_u| - |G_u^1|  - \sum_{i\in B_c \setminus B_0} v_i(\x^*)\ .
\end{align}
The first inequality holds since $v_{ij} \le \min(c_i,p_j)$ for all $i$ and $j$ and the goods in $G_c$ and $G_u^1$ are allocated to the agents in $B_c$ and $B_u^1$, respectively. The second line follows since the total money flow into 
the goods is $\sum_{j \in G_c} p_j^a + \abs{G_u^1} + \sum_{j \in G \setminus (G_c \cup G_u^1)} p_j$ and the total money flow out of the agents is $\sum_{i \in B_c \setminus B_0} c_i + \abs{B_u}$. The last line follows from $\sum_{i\in B_c \setminus B_0} v_i(\x^*) \le \sum_{j\in G_c} p_j^a$. This holds since $v_i(\x^*) \le c_i \le 1$ for every $i \in B_c \setminus B_0$ and hence any good $j \in G_c$ can contribute at most $\min(c_i,p_j) \le \min(1,p_j) = p_j^a$ to $v_i(\x^*)$. 

We now take a fractional improvement step and relax the integrality condition on $\x^*$ for buyers in $B \setminus (B_u^1 \cup B_0)$. 
We take the goods assigned to $B \setminus (B_u^1 \cup B_0)$ and redistribute them fractionally among these buyers. However, we require that the fractional solution respects the upper bound~\eqref{eq:nubBu}. We denote by $\tilde{\x}$ the best solution obtained in this improvement step. Note that the Nash social welfare can only increase. Further, note that $\tilde{\x}$ is integral for the buyers in $B_u^1$. 

Since $\tilde{\x}$ satisfies~\eqref{eq:nubBu}, 
\begin{equation}\label{eqn:t2}
\sum_{i \in B_u \setminus B_u^1} v_i(\tilde{\x})  \le |B_u| - |G_u^1| + \sum_{i\in B_c \setminus B_0} c_i - \sum_{i\in B_c \setminus B_0} v_i(\tilde{\x})\ .
\end{equation}
We want to bound
\begin{align*}
 \prod_{i\in B \setminus B_0} v_i(\tilde{\x}) &= \prod_{i\in B_c \setminus B_0} v_i(\tilde{\x}) \cdot \prod_{i\in B_u^1 } v_i(\tilde{\x}) \cdot \prod_{i\in B_u \setminus B_u^1} v_i(\tilde{\x}) \\
       &\le \prod_{i\in B_c \setminus B_0} v_i(\tilde{\x}) \cdot \prod_{i\in B_u^1 } v_i(\tilde{\x}) \cdot \left(\frac{ \sum_{i\in B_u \setminus B_u^1} v_i(\tilde{\x}) }{\abs{B_u \setminus B_u^1}} \right)^{\abs{B_u \setminus B_u^1}}.
\end{align*} 
We will show that the maximum value is obtained when each buyer $i\in B_c \setminus B_0$ gets value $c_i$, each buyer $i \in B_u^1$ gets exactly one good of $G_u^1$, i.e., $|B_u^1| = |G_u^1|$, and each buyer $i\in B_u\setminus B_u^1$ gets value $1$. This will prove the claim. 

Assume first that $\sum_{i \in B_u \setminus B_u^1} v_i(\tilde{\x}) > \abs{B_u \setminus B_u^1}$. Then there must be a buyer $i\in B_c \setminus B_0$ that gets value less than $c_i$ and hence less than one and a buyer $i' \in B_u \setminus B_u^1$ that gets value more than one. Since $\tilde{\x}$ is allowed to be fractional for buyers in $B_c \setminus B_0$ and $B_u\setminus B_u^1$, we can reallocate some amount of good from $i'$ to $i$. This increases the Nash social welfare, a contradiction.
We now have $\sum_{i \in B_u \setminus B_u^1} v_i(\tilde{\x}) \le  \abs{B_u \setminus B_u^1}$. 

Each buyer in $i \in B_u^1$ gets at least one good $j$ with price $p_j > 1$ in $\tilde{\x}$. Since $v_{ij} = \min(c_i, p_j) > 1$, we have $v_i(\tilde{\x}) > 1$ for all $i \in B_u^1$. Suppose now that at $\tilde{\x}$ we have $|B_u^1| < |G_u^1|$, i.e., a buyer $i'\in B_u^1$ gets at least two goods of $G_u^1$. Then \eqref{eqn:t2} implies that either
$ \sum_{i \in B_u \setminus B_u^1} v_i(\tilde{\x}) < \abs{B_u \setminus B_u^1}$ and hence there is an $i \in B_u \setminus B_u^1$ with $v_i(\tilde{\x}) < 1$ or there is an 
$i\in B_c \setminus B_0$ that gets value less than $c_i$ and hence less than one or both. Since $\tilde{\x}$ is allowed to be fractional for buyers in $B_c \setminus B_0$ and $B_u\setminus B_u^1$, we can reallocate one entire good from $i'$ to $i$. This increases the Nash social welfare, a contradiction.
We now have $\abs{B_u^1} = \abs{G_u^1}$, i.e., the goods of price higher than one are in one-to-one correspondence to the buyers in $B_u^1$. Thus
\[ \prod_{i\in B \setminus B_0} v_i(\tilde{\x}) = \prod_{i\in B_c \setminus B_0} c_i  \cdot \prod_{j: p_j > 1 } p_j \cdot \left( 1 \right)^{\abs{B_u \setminus B_u^1}} \ .\qedhere\]
\end{proof}

\subsubsection{Rounding Equilibria}
\label{sec:algoNSW}

In this section, we give an algorithm to round a fractional allocation of a thrifty and modest equilibrium $(\x,\p)$ to an integral one. Without loss of generality, we may assume that the allocation graph $(B\cup G, E)$ with $E = \{\{i,j\} \in B \times G \mid x_{ij} > 0\}$ is a forest~\cite{Orlin10,DuanGM16}. In the following, we only discuss how to round the trees in $(B \setminus B_0) \times (G \setminus G_0)$. For trees in $B_0 \times G_0$, the rounding and the analysis are very similar, but independent of prices and slightly simpler (see Appendix~\ref{app:roundPrice0}). Consider the following procedure:
\medskip

\begin{description}
\item[Preprocessing:] 
It consists of three substeps.
\begin{compactenum}[(a)]
\item For each tree component of the allocation graph, assign some agent to be a root node. 
\item For every good $j$ keep at most one child agent. This child-agent $i$ must buy the largest amount of $j$ among the child agents (ties are broken arbitrarily) and must have an active budget which is less than twice the price of $j$, i.e., $m_i^a/2 < p_j$. In other words, child agent $i$ is cut off from good $j$ if a sibling buys more of good $j$ (ties are broken arbitrarily) or if $p_j \le m_i^a/2$. Note that if a sibling buys more of good $j$, it also spends more on good $j$. 
\item  Agents whose connection to their parent-good is severed in step (b) become roots.
\end{compactenum}

\item[Rounding:] It consists of two substeps.
\begin{compactenum}[(a)]
\item Goods with no child agent are assigned to their parent agent.
\item For each non-trivial tree component, do the following recursively: Assign the root agent a child good $j$ that gives him the maximum value (among all children goods) in the fractional solution. 
Except in the subtree rooted at $j$, assign each good to its child agent in the remaining tree. Make the child agent of good $j$ the root node of the newly created tree. 
\end{compactenum}
\end{description}

\begin{lemma}\label{lem:tree}
After preprocessing, the valuation of each root agent $r$ is at least $v_r(\x)/2$. For all other agents $i$ the valuation is at least $v_i(\x)$. If good $j$ has child agents and $p_j > 1$, then $j$ keeps a child agent. 
\end{lemma}

\begin{proof}
Whenever an agent $i$ loses allocation because the connection to its parent-good $j$ is cut, a new tree component is being created and $i$ becomes its root node. Since $v_{ij} \le c_i, \forall (i,j)$, each capped agent needs to buy in total at least one unit of goods, and each uncapped agent spends his entire budget. 
If $i$ is cut from $j$, then $x_{ij} \le 1/2$ or $p_j \le m_i^a/2$.  In the former case, we distinguish cases. Since the total spending on good $j$ is at most one and a sibling of $i$ spends at least as much on $j$ as $i$, $i$ spends at most $1/2$ on $j$. Thus, if $i$ is uncapped, it receives at most half of its utility from $j$. 
If $i$ is capped, $v_{ij}x_{ij} \le c_i/2 = v_i(\x)/2$, i.e., $i$ receives at most half of its utility from $j$.  
In the latter case, $v_{ij} x_{ij} \le \min\{c_i,p_j\} \le p_j \le m_i^a/2 = v_i(\x)/2$, i.e., $i$ receives at most half of its utility from $j$.   

For a good $j$ with child agents and $p_j > 1$, the child agent $i$ that buys most of $j$ is kept as a child since $p_j > 1/2 = m_i^a/2$. 
\end{proof}

\begin{lemma}\label{lem:trivialtrees} 
After step (a) of rounding, each tree component $T$ has $k_T + 1$ agents and $k_T$ goods for some $k_T \ge 0$. Suppose agent $i$ in $T$ is assigned a good $j$ with $p_j > 1$ during step (a) of rounding. Then $i$ spends all its money on $j$, and $i$ is the only agent that spends money on $j$. Moreover, the valuation of $i$ after rounding is $p_j$, and $B_c \cap T = \emptyset$.
\end{lemma}
\begin{proof} The first part is straightforward since after step (a) of rounding, every remaining good has exactly one parent agent and one child agent. For the second part, consider any good $j$ with price $p_j > 1$. If $j$ has children-agents, the one that spends most on $j$ stays as a child. Thus $j$ is not assigned during preprocessing. If $j$ is a leaf of the initial forest, only its parent agent spends money on it, call it $i$.  Since the money inflow into $j$ is one and $i$ has only one unit to spend, only $i$ spends on $j$ and $i$ spends one unit on $j$. Thus $i$ is not capped in the equilibrium and $v_{i j} = p_j$ since $j$ is an MBB-good for $i$ and hence $p_j \le c_i$ by (\ref{eq:ubMatch}). Thus the valuation of $i$ after rounding is $p_j$ and $B_c \cap T = \emptyset$. 
\end{proof}

\begin{lemma}\label{lem:half}
After rounding, each agent $i$ that is assigned its parent good obtains a valuation of at least $v_i(\x)/2$. 
\end{lemma}

\begin{proof}
Consider any good $j$ in the tree in the rounding step. Since $j$ was not assigned to its parent agent during preprocessing, its price is at least half of the active budget of its child agent, i.e., $p_j \ge m_i^a/2$. Since $j$ is MBB for $i$, $v_{ij} = p_j$ and hence from this good the child-agent obtains a valuation of at least half of the valuation in the equilibrium.
\end{proof}

Consider a tree $T$ at the beginning of the step (b) of rounding with $k_T+1$ agents and $k_T$ goods. Assume $k_T \ge 1$ first. Let $a_1, g_1, a_2, g_2,$ $\dots, a_l, g_\ell, a_{\ell+1}$ be the \emph{recursion path} in $T$ starting from the root agent $a_1$ and ending at the leaf agent $a_{\ell +1}$ such that $a_1,\ldots, a_{\ell+1}$ became root agents of the trees formed recursively during the rounding step, and good $g_i$ is assigned to $a_i$ in this process, for $1\le i\le \ell$. Note that $a_{\ell + 1}$ is not assigned any good in step (b) of rounding. However, as the proof of the following Lemma shows, it must have been assigned some good during step (a) of rounding. We denote by $k_i$ the number of children for agent $a_i$, for $1 \le i \le \ell$. If $k_T = 0$, then $\ell = 0$ and $a_1 = a_{\ell +1}$ is the root of a tree containing no goods after step (a) of rounding. 

\begin{lemma}\label{lem:treeb}
The product of the valuations of agents in $T$ in the rounded solution is at least
\[ 
  \left(\frac{1}{2}\right)^{k_T-\ell+1} \cdot \frac{1}{k_1\cdots k_\ell} \cdot \prod_{i\in T\cap B_c} c_i \prod_{j \in T: p_j > 1} p_j \enspace.
\]
\end{lemma}

\begin{proof} Let $\bc_i = \min\{1, c_i\}, \forall i\in B$. 

We first deal with the case $k_T = 0$. Then $\ell = 0$. If a good $j$ of price $p_j > 1$ is assigned to $a_1$ during step (a) of rounding, then the valuation of $a_1$ after rounding is $p_j$ and $B_c \cap T = \emptyset$ by Lemma~\ref{lem:trivialtrees}. This establishes the claim even without the leading factor $1/2$. If all goods assigned to $a_1$ during step (a) of rounding have price at most one then $\{j \in T : p_j > 1\} = \emptyset$ and 
$T \cap B_c \subseteq \{a_1\}$. Moreover, the value of $a_1$ after preprocessing is at least $v_{a_1}(\x)/2 = \bc_{a_1}/2$. Rounding does not decrease the value.

We turn to the case $k_T \ge 1$. Let $q_i = x_{a_i,g_i} > 0$ be the amount of good $g_i$ bought by agent $a_i$ in the equilibrium, for $1 \le i \le l$. Then $a_i$ spends $q_i p_i$ on good $g_i$. 

In the market equilibrium, the root agent $a_1$ receives at least half of its valuation from its children. Thus $q_1 p_1 \ge \bc_1/(2k_1)$. 

We next show that agent $a_i$, $2 \le i \le \ell+1$, receives at least value $q_{i-1}\max(\bc_i, p_{i-1})$ from its children in the market equilibrium. Agent $i$ can spend at most $\bc_i - q_{i-1} p_{i-1}$ on good $i - 1$. Thus it must spend $q_{i-1} p_{i-1}$ on its children in the market equilibrium. This establishes the claim if $\bc_i \le p_{i-1}$. So assume $p_{i-1} < \bc_i$. Agent $i$ can receive at most a fraction $1 - q_{i-1}$ of good $i-1$. Hence the value it receives from this good is at most $(1 - q_{i-1})p_{i-1} \le (1 - q_{i-1})\bc_i$. Thus it must receive value at least $q_{i-1}\bc_i$ from its children goods. 

$q_i p_i \ge q_{i-1}\max(\bc_i, p_{i-1})/k_i$ for $2 \le i \le \ell$, since agent $a_i$ spends $q_i p_i$ on good $g_i$ and this is at least a fraction $1/k_i$ of what it spends totally on its children. 

The product of the valuations of $a_1$ to $a_{\ell + 1}$ in the rounded solution is at least $p_1 \ldots p_\ell \cdot q_\ell \max(\bc_\ell, p_\ell)$. This holds since $g_i$ is assigned to $a_i$ for $1 \le i \le \ell$ and $a_{\ell + 1}$ receives a value at least $q_\ell \max(\bc_\ell, p_\ell)$ from its children in the market equilibrium. Since these children are assigned to $a_{\ell + 1}$ during step (a) of rounding, it receives at least this value in the rounded solution. 

Combining the arguments above we obtain 
\begin{align*}
p_1 \cdots p_\ell \cdot &q_\ell \max(\bc_\ell, p_\ell)\\
&\ge  \frac{\bc_1}{2 q_1 k_1}\frac{q_1 \max(\bc_2,p_{1})}{q_2 k_2}\cdots \frac{q_{\ell - 1}\max(\bc_\ell,p_{\ell-1})}{q_\ell k_\ell}\cdot q_\ell \max(\bc_{\ell+1},p_\ell) \\
&=\frac{1}{2} \left(\frac{1}{k_1\dots k_{l}}\right) \bc_1 \cdot \prod_{2\le i \le \ell+1} \max({\bc_i,p_{i-1}})\\
&\ge \frac{1}{2} \left(\frac{1}{k_1\dots k_{l}}\right) \prod_{1 \le i \le \ell+1} \bc_i \cdot \prod_{1 \le i \le \ell; p_i > 1} p_i,
\end{align*}
where the last inequality follows from $\max(\bc_i,p_{i-1}) \ge \bc_i \cdot \max(1,p_{i-1})$ for all $i$. 

Each of the remaining $k_T- \ell$ agents in $T$ get a value at least $\max(v_i(\x)/2, p)$, where $p$ is the price of the parent-good. Finally, since at most one good is assigned to each agent during the rounding step, each capped good is assigned to a separate agent, the product of the valuations of agents in $T$ in the rounded solution is at least
\[
    \left(\frac{1}{2}\right)^{k-l+1}\left(\frac{1}{k_1\dots k_{l}}\right) \prod_{i\in T\cap B_c} c_i \prod_{j \in T: p_j > 1} p_j\enspace. \qedhere
\]
\end{proof}

\begin{theorem}
The rounding procedure gives a $2e^{1/2e}$-approximation for the optimal Nash social welfare with budget-additive valuations. Note that $2e^{1/2e} < 2.404$.
\end{theorem}

\begin{proof}
Suppose there are trees $T^1, T^2, \dots, T^a$ at the beginning of the rounding. Let $k^i + 1$ and $k^i$ be the number of agents and goods in tree $T^i$, respectively. Let $l^i+1$ be the number of agents on the path in $T^i$ traced during the rounding step, and let $k^i_1,\ldots,k^i_{l_i}$ be the degrees of the number of children goods for agents along that path.

The bound in Lemma \ref{lem:treeb} for trees $T \subseteq (B \setminus B_0) \times (G \setminus G_0)$ can also be obtained for our rounding of trees $T \subseteq B_0 \times G_0$ (Lemma~\ref{lem:treeb0} in the Appendix). Thus, the Nash social welfare of the rounded solution is at least 
\begin{align*}
& \left(\left(\frac{1}{2}\right)^{\sum_{i=1}^a (k^i-l^i+1)} 
\left( \frac{1}{k^1_1\dots k^1_{l^1} k^2_1\dots k^2_{l^2} \dots k^a_1\dots k^a_{l^a}}\right)
\prod_{i\in B_c} c_i \prod_{j : p_j > 1} p_j\right)^{1/n} \\
= \quad &\frac{1}{2} \cdot 2^{\sum_{i=1}^a l^i/n} \cdot \left(\frac{1}{\prod_{i=1}^{a} \prod_{j=1}^{l^i}k^i_j}\right)^{1/n} \left(\prod_{i\in B_c} c_i \prod_{j : p_j > 1} p_j\right)^{1/n} \\
\ge \quad &\frac{1}{2}\left(\frac{2\sum_{i=1}^a l^i}{\sum_{i=1}^a\sum_{j=1}^{l^i} k^i_j}\right)^{\sum_{i=1}^a l^i/n} \left(\prod_{i\in B_c} c_i \prod_{j : p_j > 1} p_j\right)^{1/n} 
\ge \quad \frac{1}{2e^{1/2e}}\left(\prod_{i\in B_c} c_i \prod_{j : p_j > 1} p_j\right)^{1/n}\enspace. 
\end{align*}
The first equation uses $\sum_{i} (k_i + 1) = n$. The subsequent inequality follows from the standard relation of arithmetic and geometric mean applied to the set of all $k_j^i$, i.e., $\left(\prod_i \prod_j k^i_j\right)^{1/\sum_i l^i} \le \sum_i \sum_j k^i_j/ \sum_i l^i$. The last inequality uses $\sum_{i=1}^a \sum_{j=1}^{l^i} k_j^i \le n$ and the fact that $(2x)^x$ is minimum at $x = 1/2e$. 
\end{proof}

\subsubsection{Rounding Equilibria of Perturbed Markets}
\label{sec:perturbNSW}

Given a parameter $\eps' > 0$, our FPTAS in Section~\ref{sec:FPTAS} computes an exact equilibrium for a perturbed market, which results when agents have perturbed valuations $\tilde{v}_i(\x) = \min\left(c_i, \sum_{j} \tilde{v}_{ij}x_{ij}\right)$ with the same caps $c_i$ and $\tilde{v}_{ij} \ge v_{ij} \ge \tilde{v}_{ij}/(1+\eps')$. Suppose we apply our rounding algorithm to the exact equilibrium for $\tilde{v}$. It obtains an allocation $S$ such that
\begin{align*}
    \left(\prod_{i} v_i(\x_i^S)\right)^{1/n} &\ge \frac{1}{(1+\eps')} \left(\prod_{i} \tilde{v}_i(\x_i^S)\right)^{1/n} \\
                          &\ge \frac{1}{(1+\eps')} \cdot \frac{1}{2e^{1/2e}} \left(\prod_{i} \tilde{v}_i(\x^*)\right)^{1/n}\\
                          &\ge \frac{1}{(1+\eps') \cdot 2e^{1/2e}} \cdot \left(\prod_{i} v_i(\x^*)\right)^{1/n} \enspace.
\end{align*}
If we apply the FPTAS with $\eps'$, then this yields an approximation ratio of at most $2e^{1/2e} + \eps$ for $\eps = 2e^{1/(2e)}\eps'$. We summarize our main result:

\begin{corollary}
  For every $\eps > 0$ there is an algorithm with running time polynomial in $n$, $m$, $\log \max_{i,j}\{v_{ij}, c_i\}$, and $1/\eps$ that computes an allocation which represents a $(2e^{1/2e}+\eps)$-approximation for the optimal Nash social welfare.
\end{corollary}

\newcommand{\xbrack}[1]{\langle #1 \rangle}

\subsection{Hardness of Approximation}
\label{sec:LB}

In this section, we provide a result on the hardness of approximation of the maximum Nash social welfare with additive valuations. The best previous bound was a factor of 1.00008~\cite{Lee17}. Our improved lower bound of $\sqrt{8/7} > 1.069$ follows by adapting a construction in~\cite{ChakrabartyG10} for (sum) social welfare with budget-additive valuations.

\begin{theorem}
  \label{theo:hardness}
  For every constant $\delta > 0$, there is no $(\sqrt{8/7} - \delta)$-approximation algorithm for maximizing Nash social welfare with additive valuations unless \classP=\classNP.
\end{theorem}

For clarity, we first describe the proof for budget-additive valuations with caps. Subsequently, we show how to drop the assumption of caps and apply the proof even for additive valuations.

\begin{lemma}
  \label{lem:hardness}
  For every constant $\delta > 0$, there is no $(\sqrt{8/7} - \delta)$-approximation algorithm for maximizing Nash social welfare with budget-additive valuations unless \classP=\classNP.
\end{lemma}

\begin{proof}
  Chakrabarty and Goel~\cite{ChakrabartyG10} show hardness for (sum) social welfare by reducing from MAX-E3-LIN-2. An instance of this problem consists of $n$ variables and $m$ linear equations over GF(2). Each equation consists of 3 distinct variables. For the Nash social welfare objective, we require slightly more control over the behavior of the optimal assignments. Therefore, we consider the stronger problem variant Ek-OCC-MAX-E3-LIN-2, in which each variable occurs exactly $k$ times in the equations.

  \begin{theorem}[\cite{ChlebikC03}]
    For every constant $\varepsilon \in (0,\frac 14)$ there is a constant $k(\varepsilon)$ and a class of instances of Ek-OCC-MAX-E3-LIN-2 with $k \ge k(\varepsilon)$, for which we cannot decide if the optimal variable assignment fulfills more than $(1-\varepsilon)m$ equations or less than $(1/2 + \varepsilon)m$ equations, unless \classP=\classNP.
  \end{theorem}

  Our reduction follows the construction in~\cite{ChakrabartyG10}. We only sketch the main properties here. For more details see~\cite[Section 4]{ChakrabartyG10}.

  For each variable $x_i$ we introduce two agents $\xbrack{x_i : 0}$ and $\xbrack{x_i : 1}$. Each of these agents has a cap of $c_i = 4k$, where $k$ is the number of occurrences of $x_i$ in the equations. Since in Ek-OCC-MAX-E3-LIN-2 every variable occurs exactly $k$ times, we have $c_i = 4k$ for all agents. Moreover, for each variable $x_i$ there is a \emph{switch item}. The switch item has value $4k$ for agents $\xbrack{x_i : 0}$ and $\xbrack{x_i : 1}$, and value 0 for every other agent. It serves to capture the assignment of the variable -- if $x_i$ is set to $x_i = 1$, the switch item is given to $\xbrack{x_i : 0}$ (for $x_i = 0$, the switch item goes to $\xbrack{x_i : 1}$). When given a switch item, an agent cannot generate value for any additional equation items defined as follows.

  For each equation $x_i + x_j + x_k = \alpha$ with $\alpha \in \{0,1\}$, we introduce 4 classes of equation items -- one class for each satisfying assignment. In particular, we get class $\xbrack{x_i : \alpha; x_j : \alpha; x_k : \alpha}$ as well as classes $\xbrack{x_i :  \bar{\alpha}, x_j : \bar{\alpha}, x_k : \alpha}$, $\xbrack{x_i :  \bar{\alpha}, x_j : \alpha, x_k : \bar{\alpha}}$ and $\xbrack{x_i : \alpha, x_j : \bar{\alpha}, x_k : \bar{\alpha}}$. For each of these classes, we introduce three items. Hence, for each equation we introduce 12 items in total. An item $\xbrack{<x_i : \alpha_i, x_j : \alpha_j, x_k : \alpha_k}$ has a value of 1 for the three agents $\xbrack{x_i : \alpha_i}$, $\xbrack{x_j : \alpha_j}$, and $\xbrack{x_k : \alpha_k}$, and value 0 for every other agent.

  It is easy to see that w.l.o.g.\ every optimal assignment of items to agents assigns all switch items. Hence, every optimal assignment yields some variable assignment for the underlying instance of Ek-OCC-MAX-E3-LIN-2. 

  Consider an equation $x_i + x_j + x_k = \alpha$ that becomes satisfied by setting the variables $(x_i, x_j, x_k) = (\alpha_i, \alpha_j, \alpha_k)$. Then none of the agents $\xbrack{x_i : \alpha_i}$,  $\xbrack{x_j : \alpha_j}$, and $\xbrack{x_k : \alpha_k}$ gets a switch item, and we can assign exactly 4 equation items to each of these agents (for details see~\cite{ChakrabartyG10}). Hence, all 12 equation items generate additional value. In particular, it follows that if $x_i$ is involved in a satisfied equation, one of its agent gets a switch item, and the other one can receive at least 3 equation items.
  
  Consider an equation $x_i + x_j + x_k = \alpha$ that becomes unsatisfied by setting the variables $(x_i, x_j, x_k) = (\alpha_i, \alpha_j, \alpha_k)$. Then for one class of equation items, all agents that value these items have already received switch items (for details see~\cite{ChakrabartyG10}). This class of items cannot generate additional value. Hence, at most 9 equation items generate additional value. They can assigned to the agents that did not receive switch items such that each agent receives 3 items. In particular, it follows that if $x_i$ is involved in an unsatisfied equation, one of its agents gets a switch item, and the other one can receive at least 3 equation items. Hence, we can ensure that in every optimal solution the overall Nash social welfare is never 0.
  
  We now derive a lower bound on the optimal Nash social welfare when $(1-\varepsilon)m$ equations can be satisfied. In this case, we obtain value $4k$ for $n$ agents that receive the switch items. Moreover, we get an additional total value of $12m(1-\varepsilon) + 9m\varepsilon$ generated by the equation items. Note that $m = kn/3$. We strive to lower bound the Nash social welfare of an optimal assignment in this case. For this, it suffices to consider the assignment indicated above -- for each satisfied equation, all incident agents without switch items get 4 equation items. For each unsatisfied equation, all incident agents without switch items get 3 equation items. To obtain a lower bound on the Nash social welfare, we assume a value of $4k$ for a maximum of $n(1-\varepsilon)$ agents, while the others get a value of $3k$. Therefore, when an assignment of items to agents generates Nash social welfare of more than
  \[
    \left( (4k)^n \cdot (4k)^{n(1-\varepsilon)}\cdot (3k)^{n\varepsilon}\right)^{-2n} = k \cdot 4^{\frac 12}\cdot 4^{\frac 12} \cdot (3/4)^{\frac{\varepsilon}{2}}\enspace,
  \]
  we take this as an indicator that at least $m(1-\varepsilon)$ equations can be fulfilled.
 
  In contrast, now suppose only $(1/2 + \varepsilon)m$ equations can be fulfilled. In this case, we obtain value $4k$ for $n$ agents that receive the switch items. Moreover, we get an additional total value of at most $12m(1/2 + \varepsilon) + 9m(1/2 - \varepsilon) = 10.5m + 3\varepsilon m$ generated by the equation items. We strive to upper bound the Nash social welfare of such an assignment. For this, we assume that all agents that do not receive a switch item get an equal share of the value generated by equation items, i.e., a share of $3.5k + k\varepsilon$. Therefore, when an assignment of items to agents generates Nash social welfare of less than
  \[
    \left( (4k)^n \cdot \left(k(3.5 + \varepsilon)\right)^n\right)^{-2n} = k \cdot 4^{\frac 12} \cdot \left(3.5 + \varepsilon \right)^{\frac 12}\enspace,
  \]
  we take this as an indicator that at most $m(1/2 + \varepsilon)$ equations can be fulfilled. 

  Hence, if we can approximate the optimal Nash social welfare by at most a factor of
  \[
    \frac{4^{\frac 12} \cdot (3/4)^{\frac{\varepsilon}{2}}}{(3.5 + \varepsilon)^{\frac{1}{2}}} = \left(\frac{4 \cdot (3/4)^{\varepsilon}}{3.5+\varepsilon}\right)^{\frac 12}\enspace,
  \]
  we can decide whether the instance of Ek-OCC-MAX-E3-LIN-2 has an optimal assignment with at least $m(1-\varepsilon)$ or at most $m(1/2+\varepsilon)$ satisfied equations. This shows that we cannot approximate Nash social welfare with a factor of $\sqrt{8/7} > 1.069$ unless \classP=\classNP.
\end{proof}

Having established the result for budget-additive valuations, we now show how to adjust the construction for additive valuations.

\begin{proof}[Proof of Theorem~\ref{theo:hardness}]
  We use the same construction as in Lemma~\ref{lem:hardness} with the following adjustments. Switch items have a large value $M \gg 4k$ for the respective agents. All agent valuations are additive and have no caps (i.e., all $c_i = \infty$). 

  First, we again establish the lower bound on the optimal Nash social welfare when  $(1-\varepsilon)m$ equations can be satisfied. Then $12m(1-\varepsilon) + 9m\varepsilon$ equation items can be given to the agents without switch items. The remaining $3m\varepsilon$ equation could be assigned to agents with switch items. Instead, to construct a lower bound, we simply drop these items from consideration. Therefore, when an assignment of items to agents generates Nash social welfare of more than
  \[
    \left( M^n \cdot (4k)^{n(1-\varepsilon)}\cdot 3k^{n\varepsilon}\right)^{-2n} = k^{\frac 12} \cdot M^{\frac 12} \cdot 4^{\frac 12}\cdot (3/4)^{\frac{\varepsilon}{2}}
  \]
  we take this as an indicator that at least $m(1-\varepsilon)$ equations can be fulfilled.

  Now suppose only $(1/2 + \varepsilon)m$ equations can be fulfilled. To construct an upper bound on the optimal Nash social welfare, we apply the reasoning above. $n$ agents receive switch items according to an optimal variable assignment. The $12m(1/2 + \varepsilon) + 9m(1/2 - \varepsilon) = 10.5m + 3\varepsilon m$ equation items are assigned in equal shares to agents without switch items. The remaining $3m(1/2 - \varepsilon)$ items are assigned equal shares to agents with switch items. Therefore, when an assignment of items to agents generates Nash social welfare of less than
  \[
    \left(  \left(M+k\left(\frac 12 - \varepsilon\right)\right)^n \cdot \left(k(3.5 + \varepsilon)\right)^n\right)^{-2n} =  k^{\frac 12} \cdot \left(M+k\left(\frac 12 - \varepsilon\right)\right)^{\frac 12} \cdot (3.5 + \varepsilon)^{\frac 12}
  \]
  we take this as an indicator that at most $m(1/2 + \varepsilon)$ equations can be fulfilled. 

  Hence, if we can approximate the optimal Nash social welfare by at most a factor of
  \[
    \left( \frac{M}{M+k\left(\frac 12 - \varepsilon\right)}\right)^{\frac 12} \cdot \left( \frac{4 \cdot (3/4)^{\varepsilon}}{3.5 + \varepsilon} \right)^{\frac 12}
  \]
  we can decide whether the instance of Ek-OCC-MAX-E3-LIN-2 has an optimal assignment with at least $m(1-\varepsilon)$ or at most $m(1/2+\varepsilon)$ satisfied equations. The second fraction clearly grows to $\sqrt{8/7}$ as $\varepsilon \to 1$. For a fixed number $M$, the first fraction decreases, since $k = k(\varepsilon)$ increases with decreasing $\varepsilon$. However, we can choose a number $M = o(k)$ since $k$ is the number of occurrences of a single variable and, thus, the input size is polynomial in $k$. For example, with $M = k^k$ the first term approaches 1 rapidly as $\varepsilon \to 0$ (and $k \to \infty$). This shows that we cannot approximate Nash social welfare with a factor of $\sqrt{8/7} > 1.069$ unless \classP=\classNP, even for additive valuations.
\end{proof}

\section{Future Directions}
\label{sec:Future}

There are many interesting questions arising from our work. Equilibria in linear Fisher markets with both earning and utility limits turn out to have intriguing structure. Although equilibrium may not always exist in these markets, we showed that it always exists when the market satisfies the money clearing condition, i.e., money clearing is a (polynomial-time checkable) sufficiency condition for existence. While existence is guaranteed in this case, the set of equilibria is non-convex. Nevertheless, we managed to obtain an FPTAS to compute an approximate equilibrium. Beyond money-clearing markets, however, even deciding existence is not well-understood -- is it \classNP-hard, or can it be tightly characterized by a simple condition that can be checked in polynomial time? In addition to deciding existence, can we efficiently \emph{find} exact equilibria in general (if they exist) or in money-clearing markets? Is the problem of finding equilibria for money-clearing markets in the class \classCLS\ (Continuous Local Search)~\cite{DaskalakisP11}? 

We showed that an approximate equilibrium can be rounded to obtain a constant-factor approximation of Nash social welfare for allocating indivisible goods among agents with budget-additive valuations. This shows that in order to obtain a constant-factor approximation algorithm for the Nash social welfare problem, we only need a constant-factor approximate equilibrium of an appropriate Fisher market. It would be interesting to see if Fisher markets with limits can yield good approximation algorithms for the Nash social welfare problem also beyond budget-additive and additive-separable concave valuations. 

\appendix

\newcommand{\MM}{\mbox{\boldmath $M$}}
\newcommand{\yy}{\mbox{\boldmath $y$}}
\newcommand{\pq}{\mbox{\boldmath $q$}}
\newcommand{\CP}{\mbox{${\cal P}$}}
\newcommand{\pa}{\mbox{\boldmath $a$}}
\newcommand{\one}{\mbox{\boldmath $1$}}
\newcommand{\pv}{\mbox{\boldmath $v$}}

\section{The Linear Complementarity Problem and Lemke's Algorithm}
\label{app:lcp}

Given an $n \times n$ matrix $\MM$, and a vector $\pq$, the linear complementarity problem\footnote{We refer the reader to~\cite{CottlePS92} for a comprehensive treatment of notions presented in this section.} asks for a vector $\yy$ satisfying the following
conditions:
\begin{equation}
\label{eq.a} 
\MM \yy \leq \pq,  \ \ \ \   \yy \geq 0 \ \ \ \ \mbox{and} \ \ \ \  \yy \cdot (\pq - \MM \yy) = 0.  
\end{equation}

The problem is interesting only when $\pq \not \geq 0$, since otherwise $\yy = 0$ is a trivial solution. Let us introduce slack 
variables $\pv$ to obtain the equivalent formulation.
\begin{equation}
\label{eq.b} 
 \MM \yy  + \pv = \pq, \ \ \ \  \yy \geq 0, \ \ \ \ \pv \geq 0 \ \ \ \ \mbox{and} \ \ \ \ \yy \cdot \pv = 0 .
\end{equation}

The reason for imposing non-negativity on the slack variables is that the first condition in (\ref{eq.a}) implies $\pq - \MM \yy \geq
0$. Let $\CP$ be the polyhedron in $2n$ dimensional space defined by the first three conditions; we will assume that $\CP$ is
non-degenerate\footnote{A polyhedron in $n$-dimension is said to be {\em non-degenerate} if on its $d$-dimensional faces exactly $n-d$
of its constraints hold with equality. For example on vertices ($0$-dimensional face) exactly $n$ constraints hold with equality. There
are many other equivalent ways to describe this notion.}.
Under this condition, any solution to (\ref{eq.b}) will be a vertex of $\CP$, since it must satisfy $2n$ equalities. Note that the set
of solutions may be disconnected.

An ingenious idea of Lemke was to introduce a new variable and consider the system, which is called the {\em augmented LCP}:
\begin{equation}
\label{eq.c} 
 \MM \yy  + \pv -z \one  = \pq, \ \ \ \  \yy \geq 0, \ \ \ \ \pv \geq 0, \ \ \ \  z \geq 0  \ \ \ \ \mbox{and} \ \ \ \ \yy \cdot \pv = 0 .
\end{equation}

Let $\CP'$ be the polyhedron in $2n + 1$ dimensional space defined by the first four conditions of the augmented LCP; again we 
will assume that $\CP'$ is non-degenerate. Since any solution to (\ref{eq.c}) must still satisfy $2n$ equalities,
the set of solutions, say $S$, will be a subset of the one-skeleton of $\CP'$, i.e., it will consist of edges and vertices of $\CP'$.
Any solution to the original system must satisfy the additional condition $z = 0$ and hence will be a vertex of $\CP'$.

Now $S$ turns out to have some nice properties. Any point of $S$ is {\em fully labeled} in the sense that for each $i$, $y_i = 0$ or
$v_i = 0$.\footnote{These are also known as {\em almost complementary solutions} in the literature.}  We will say that a point of $S$ {\em has double label i} if $y_i = 0$ and $v_i = 0$ are both
satisfied at this point. Clearly, such a point will be a vertex of $\CP'$ and it will have only one double label. 
Since there are exactly two ways of relaxing this double label, this vertex must have exactly two edges of $S$ incident at it.
Clearly, a solution to the original system (i.e., satisfying $z = 0$) will be a vertex of $\CP'$ that does not have a double label. 
On relaxing $z=0$, we get the unique edge of $S$ incident at this vertex.

As a result of these observations, it follows that $S$ consists of paths and cycles.  Of these paths, Lemke's algorithm explores a
special one.  An unbounded edge of $S$ such that the vertex of $\CP'$ it is incident on has $z > 0$ is called a {\em ray}.  Among the
rays, one is special -- the one on which $\yy = 0$. This is called the {\em primary ray} and the rest are called {\em secondary rays}.
Now Lemke's algorithm explores, via pivoting, the path starting with the primary ray. This path must end either in a vertex satisfying
$z = 0$, i.e., a solution to the original system, or a secondary ray. In the latter case, the algorithm is unsuccessful in finding a
solution to the original system; in particular, the original system may not have a solution.

{\bf Remark:}  Observe that $z \one$ can be replaced by $z \pa$, where vector $\pa$ has a 1 in each row in which $\pq$ is negative and
has either a 0 or a 1 in the remaining rows, without changing its role; in our algorithm, we will set a row of $\pa$ to 1 if and only
if the corresponding row of $\pq$ is negative.  As mentioned above, if $\pq$ has no negative components, (\ref{eq.a}) has the trivial
solution $\yy = 0$. Additionally, in this case Lemke's algorithm cannot be used for finding a non-trivial solution, since it is simply
not applicable. However, Lemke-Howson scheme is applicable for such a case; it follows a complementary path in the original
polyhedron (\ref{eq.b}) starting at $\yy=0$, and guarantees termination at a non-trivial solution if the polyhedron is
bounded.

\section{Rounding Trees with Zero Price Goods}
\label{app:roundPrice0}

In this section, we give an algorithm to round trees $T_0 \subseteq B_0 \times G_0$ of the equilibrium  $(\x,\p)$ to an integral allocation. Recall that in such trees, all goods have price $p_j = 0$ and all buyers reach their cap $c_i$. Consider the following procedure which is similar to the procedure in Section \ref{sec:algoNSW}. It uses only the allocation $\x$ and does not rely on prices. In particular, the only price-based assignment rule is in the preprocessing step, and it can be replaced here with an equivalent, more direct criterion:

\begin{description}
\item[Preprocessing:] 
It consists of three substeps.
\begin{compactenum}[(a)]
\item For each zero-price tree component of the allocation graph, assign some agent to be a root node. 
\item For every good $j$ keep at most one child agent. This child agent $i$ must buy the largest amount of $j$ among the child agents (ties are broken arbitrarily) and must receive a utility that is at least half of the total utility, i.e., $u_{ij}x_{ij} > c_i/2$. In other words, child agent $i$ is cut off from good $j$ if a sibling buys more of good $j$ (ties are broken arbitrarily) or if $u_{ij}x_{ij} \le c_i/2$. 
\item  Agents whose connection to their parent good is severed in step (b) become roots.
\end{compactenum}

\item[Rounding:] It consists of two substeps.
\begin{compactenum}[(a)]
\item Goods with no child agent are assigned to their parent agent.
\item For each non-trivial zero-price tree component, do the following recursively: Assign the root agent a child good $j$ that gives him the maximum value (among all children goods) in the fractional solution. 
Except in the subtree rooted at $j$, assign each good to its child agent in the remaining tree. Make the child agent of good $j$ the root node of the newly created tree. 
\end{compactenum}
\end{description}

Prices $p_j$ play a role in exactly two places in Section \ref{sec:algoNSW}. 

First, when we assign a good $j$ to the parent and drop the child agent $i$ with $p_j \le m_i^a/2$, this ensures that the valuation of agent $i$ in the newly created tree is at least half of the original valuation. We use this fact to prove Lemmas~\ref{lem:tree} and~\ref{lem:half}. In case of zero price goods, our choice of assigning a good $j$ to the parent and dropping the child agent $i$ if $u_{ij}x_{ij} \le c_i/2$ is an equivalent notion in terms of allocation. 

Second, in Lemmas~\ref{lem:trivialtrees} and~\ref{lem:treeb} we argue that a good with price more than $1$ is only assigned to an uncapped agent $i$ which gives agent $i$ at least $p_j$ amount of value. Since each agent of $B_0$ is capped, we do not need this property in Lemmas~\ref{lem:trivialtrees0} and~\ref{lem:treeb0}.

As a result, the proofs of the following lemmas are almost completely identical to the proofs of Lemmas~\ref{lem:tree}-\ref{lem:treeb}, respectively, and hence are omitted. 

\begin{lemma}\label{lem:tree0}
After preprocessing, the valuation of each root agent $r$ is at least $c_r/2$. For all other agents $i$ the valuation is at least $c_i$. 
\end{lemma}

\begin{lemma}\label{lem:trivialtrees0} 
After step (a) of rounding, each tree component $T$ has $k_T + 1$ agents and $k_T$ goods for some $k_T \ge 0$.
\end{lemma}

\begin{lemma}\label{lem:half0}
After rounding, each agent $i$ that is assigned its parent good obtains a valuation of at least $c_i/2$. 
\end{lemma}

Consider a zero-price tree $T$ at the beginning of the step (b) of rounding with $k_T+1$ agents and $k_T$ goods. Assume $k_T \ge 1$ first. Let $a_1, g_1, a_2, g_2,$ $\dots, a_l, g_\ell, a_{\ell+1}$ be the \emph{recursion path} in $T$ starting from the root agent $a_1$ and ending at the leaf agent $a_{\ell +1}$ such that $a_1,\ldots, a_{\ell+1}$ became root agents of the trees formed recursively during the rounding step, and good $g_i$ is assigned to $a_i$ in this process, for $1\le i\le \ell$. Note that $a_{\ell + 1}$ is not assigned any good in step (b) of rounding. However, it must have been assigned some good during step (a) of rounding. We denote by $k_i$ the number of children for agent $a_i$, for $1 \le i \le \ell$. If $k_T = 0$, then $\ell = 0$ and $a_1 = a_{\ell +1}$ is the root of a tree containing no goods after step (a) of rounding. 

\begin{lemma}\label{lem:treeb0}
The product of the valuations of agents in $T$ in the rounded solution is at least
\[ 
  \left(\frac{1}{2}\right)^{k_T-\ell+1} \cdot \frac{1}{k_1\cdots k_{\ell}} \cdot \prod_{i\in T} c_i \enspace.
\]
\end{lemma}

\bibliographystyle{abbrv} 
\bibliography{./literature,./martin}

\end{document}